\documentclass[pra,twocolumn, longbibliography]{revtex4-2}
\usepackage{graphicx}
\usepackage{amsthm}
\usepackage{amsfonts}
\usepackage{amssymb}
\usepackage{array}
\usepackage{amsmath}
\usepackage{verbatim} 
\usepackage{hyperref}
\usepackage{color}
\usepackage{bbold}
\usepackage{epstopdf}
\usepackage{mathtools}
\usepackage{qcircuit}
\usepackage[ruled,vlined]{algorithm2e}

\newtheorem{conjecture}{Conjecture}

\newtheorem{definition}{Definition}
\newtheorem{theorem}{Theorem}
\newtheorem{lemma}{Lemma}

\newtheorem{corollary}{Corollary}

\newcommand{\ket}[1]{\left\vert#1\right\rangle}
\newcommand{\bra}[1]{\left\langle#1\right\vert}

\def\bra#1{\langle #1|}
\def\ket#1{\left|#1 \right>}

\def\Tr{\mbox{Tr}}
\def\E{{\cal E}}

\def\Cl{{\rm Cl}}

\def\bra#1{\langle #1|}
\def\ket#1{\left|#1 \right>}

\def\Tr{\mbox{Tr}}

\begin{document}
\title{Testing nonstabilizerness only with stabilizer states}
\author{Hyukjoon Kwon}
\email{hjkwon@kias.re.kr}
\affiliation{School of Computational Sciences, Korea Institute for Advanced Study, Seoul 02455, Korea}
\begin{abstract}
The stabilizer formalism plays a central role in quantum information processing and quantum computing. Since stabilizer states and operations can be efficiently simulated classically and fault-tolerantly implemented in quantum error-correcting codes, quantum states and operations beyond the stabilizer framework, characterized by nonstabilizerness or magic, naturally emerge as resources for quantum computation. Here, we demonstrate that a quantum-information task involving only stabilizer states can reveal a fundamental limitation of stabilizer operations. Specifically, we construct a set of mutually orthogonal stabilizer states that cannot be perfectly distinguished using stabilizer operations and extend this construction to an arbitrary number of qubits. Our results provide an efficient test of nonstabilizerness without requiring the direct use of resourceful states or operations nor relying on computational-hardness assumptions. This nonstabilizerness test could serve as a resource-efficient benchmark for fault-tolerant quantum computers powered by magic-state injection, by providing quantitative bounds on the robustness of magic. More fundamentally, the resulting asymmetry between the preparation and discrimination of free states parallels ``nonlocality without entanglement'' in entanglement theory, revealing an unexpected connection between these two distinct quantum resource theories.
\end{abstract}
\maketitle

\section{Introduction}
Understanding the precise boundary between classical and quantum theories is not only of foundational importance but also essential to the development and implementation of quantum technologies. In quantum computing, the stabilizer formalism plays a crucial role in characterizing quantum processes that can be efficiently simulated on classical computers. The Gottesman--Knill theorem establishes that a quantum circuit composed of stabilizer-state preparations, Clifford transformations, and Pauli measurements, so-called stabilizer operations (SOs), can be efficiently simulated classically~\cite{gottesman1998heisenberg,Aaronson04improved}, whereas universal quantum computation requires nonstabilizer resources, often referred to as magic. This distinction is especially important in fault-tolerant quantum computing, where the Eastin--Knill theorem rules out realizing a universal gate set entirely through transversal operations~\cite{gottesman1997stabilizer,Eastin09restrictions}. A widely adopted route to circumvent this restriction toward universality is to implement non-Clifford operations through magic-state injection~\cite{bravyi2005universal,campbell2017roads}, in which framework, preparing the magic states often dominates the resource overhead~\cite{Howard17application}. These considerations motivate the resource theory of magic~\cite{Mari12,veitch2012negative,veitch2014resource,Howard14contextuality,Bravyi16improved,Howard17application,seddon2019quantifying,wang2019quantifying,bravyi2019simulation,Wang20efficiently,beverland2020lower,Seddon21quantifying,Leone22stabilizer}, in which nonstabilizerness is treated as the resource enabling quantum computational advantage beyond the classically simulable regime.

More generally, quantum resource theories~\cite{Chitambar2019quantum} characterize and quantify quantum resources through information-processing tasks that cannot be accomplished using only free operations. Quantum state discrimination~\cite{helstrom1969quantum,holevo1974remarks,barnett2009quantum,bae2015quantum} provides a direct setting to examine such limitations, as quantum states may become less distinguishable when the allowed operations are restricted~\cite{Bennet99quantum,divincenzo2003unextendible,fan2004distinguishability,Takagi19general,Lami21quantum,Zu24limitations}. A paradigmatic example from entanglement theory is \textit{nonlocality without entanglement}~\cite{Bennet99quantum}, in which certain sets of mutually orthogonal product states cannot be perfectly discriminated using local operations and classical communication (LOCC).  This phenomenon reveals a separation between the operationally defined class of LOCC and broader axiomatically defined classes, such as separable operations, and has inspired applications in quantum information processing, including quantum data hiding~\cite{divincenzo2002quantum,matthews2009distinguishability,lami2018ultimate,Ha24nonlocal}.

In this work, we ask a parallel question in the resource theory of magic: ``Can every ensemble of mutually orthogonal stabilizer states be perfectly discriminated by SOs?'' We show that this is not the case by constructing explicit counterexamples involving two-qubit mixed states and three-qubit pure states, and then extend the construction to multiqubit systems. The underlying obstruction is a strict information--disturbance tradeoff: the measurement backaction associated with extracting information through SOs causes an irreversible loss of distinguishability that precludes perfect discrimination. We further analyze optimal discrimination strategies using adaptive Pauli-decision trees, both with and without additional nonstabilizer ancillas.

A central consequence of this discrimination gap is a protocol for testing the nonstabilizerness of a quantum process using only stabilizer states as probes. The verifier encodes classical information in stabilizer states and asks a prover with access to the tested process to recover it. Since every SO strategy obeys a strict upper bound on its discrimination performance, exceeding this bound certifies nonstabilizerness. The protocol realizes a form of data hiding against SOs~\cite{Zu24limitations} using only free stabilizer-state inputs and provides a sample-efficient test without relying on complexity-theoretic or cryptographic assumptions. Moreover, the observed discrimination performance provides a quantitatively lower bound on the robustness of magic within the magic-state-injected Clifford-circuit model, which could be a practical method for verifying quantum computational resources and benchmarking non-Clifford gates.

At a more fundamental level, our results reveal an asymmetry between preparing stabilizer states and perfectly discriminating them using free operations. This provides a discrimination-based manifestation of the separation between operationally defined SOs and broader axiomatic classes of free operations in the resource theory of magic~\cite{heimendahl2022axiomatic}, paralleling the distinction between LOCC and separable operations in entanglement theory~\cite{Bennet99quantum}. Moreover, our result implies the existence of ensembles of stabilizer states that cannot be perfectly cloned by SOs, and we show that their optimal SO cloning fidelity equals their optimal single-copy SO discrimination probability, which is also applicable to other resource theories. Taken together, our results highlight connections among quantum resource theories, quantum foundations, and fault-tolerant quantum computation, while motivating analogous discrimination-based approaches to exploring the limitations of free operations in other quantum resource theories.

\section{Preliminaries}
\subsection{Stabilizer formalism}
For an $n$-qubit system, the Pauli group is defined as
\begin{equation}
    {\cal P}_n := \left\{ \chi P_1\otimes\cdots\otimes P_n : \chi \in \{\pm 1, \pm i\},\; P_j\in\{I,X,Y,Z\} \right\},
\end{equation}
where $I,X,Y$, and $Z$ denote the single-qubit Pauli operators. A stabilizer group ${\cal S}\subset{\cal P}_n$ is an Abelian subgroup that does not contain $-I$. A maximal stabilizer group has $2^n$ elements and can be generated by $n$ independent, mutually commuting Hermitian Pauli operators,
\begin{equation}
    {\cal S}=\langle g_1,g_2,\ldots,g_n\rangle .
\end{equation}
Based on this, stabilizer states are defined as follows:
\begin{definition}[Stabilizer state]
    A pure $n$-qubit state $\ket{\psi}$ is called a \emph{(Pauli-)stabilizer state} if there exists a maximal stabilizer group ${\cal S}=\langle g_1,\ldots,g_n\rangle\subset{\cal P}_n$ such that
    \begin{equation}
        g_j\ket{\psi}=\ket{\psi},
        \qquad j=1,\ldots,n .
    \end{equation}
Equivalently, $\ket{\psi}$ is the unique simultaneous $+1$ eigenstate of the generators of ${\cal S}$. We denote a set of $n$-qubit stabilizer states as ${\rm STAB}_n$ and often omit the subscript $n$ when the qubit number is not specified.
\end{definition}

The $n$-qubit Clifford group is the normalizer of the Pauli group in the unitary group $U(2^n)$,
$$
\Cl_n := \left\{ V\in U(2^n): V{\cal P}_nV^\dagger={\cal P}_n \right\}.
$$
Consequently, Clifford unitaries map stabilizer states to stabilizer states. Equivalently, every pure stabilizer state can be written as
$V\ket{0}^{\otimes n}$ for some $V\in\Cl_n$~\cite{gottesman1997stabilizer}.

\begin{definition}[Stabilizer operation]
    A \emph{stabilizer operation} (SO) is a quantum operation that can be implemented by a finite adaptive circuit composed of the following primitives:
    \begin{enumerate}
        \item preparing and appending ancillary qubits in computational-basis states
        \item applying Clifford unitary operations
        \item performing measurements in the computational basis
        \item discarding subsystems and classically controlling subsequent
        operations according to previous measurement outcomes.
    \end{enumerate}
\end{definition}
The above definition equivalently permits the preparation of arbitrary pure stabilizer ancillas and the measurement of arbitrary Pauli observables, since these procedures can be reduced to computational-basis preparation and measurement using Clifford unitaries. Individual measurement branches correspond to completely positive, trace-nonincreasing maps, whereas the sum over all outcomes is trace preserving.

It is important to note that any quantum circuit with a stabilizer-state input and composed entirely of stabilizer operations can be efficiently simulated on a classical computer by the Gottesman--Knill theorem~\cite{gottesman1998heisenberg, Aaronson04improved}.

\subsection{Quantum state discrimination}
In quantum state discrimination, we consider an ensemble $ {\cal E}=\{(p_k,\rho_k)\}_{k=1}^K$, where the state $\rho_k$ is prepared with probability $p_k$, and the task is to correctly identify its label $k$. For a uniformly distributed ensemble, i.e., $p_k=1/K$ for all $k$, we simply write ${\cal E}=\{\rho_k\}_{k=1}^K$, omitting the prior probabilities.

A measurement for this discrimination task is described by a positive operator-valued measure (POVM) ${\cal M}=\{M_k\}_{k=1}^K$, satisfying $M_k\geq 0$ and $\sum_{k=1}^K M_k=\mathbb{1}$. The optimal success probability under unrestricted measurements is
\begin{equation}
    p_{\rm succ}({\cal E}) = \max_{{\cal M}} \sum_{k=1}^K p_k\Tr(\rho_k M_k).
\end{equation}
For the binary discrimination task with two equiprobable states, ${\cal E}^{\rm bin}=\{\rho_0,\rho_1\}$, let $\Delta:=\rho_0-\rho_1$. The optimal success probability with a binary measurement effect ${\cal M} = \{ M, \mathbb{1}-M\}$ becomes
\begin{equation} \label{eq:bin_all}
    p_{\rm succ}({\cal E}^{\rm bin}) = \frac{1}{2} \left( 1 + \max_{{\cal M}} \Tr[ \Delta M] \right)= \frac{1}{2} \left( 1 + \frac{\| \Delta \|_1}{2} \right),
\end{equation}
where $\|A\|_1:=\Tr\sqrt{A^\dagger A}$  is the trace norm. The optimal measurement effect $M$ is the projector onto the positive eigenspace of $\Delta$. 

We now consider discrimination under stabilizer operations. Let ${\cal M}_{\rm SO} = \{ M_k \}_{k=1}^K$ denote POVM that is implementable by SO. For a general ensemble, the corresponding optimal success probability is
\begin{equation}
\begin{aligned}
    p_{\rm succ}^{\rm SO}({\cal E}) &:= \max_{{\cal M}_{\rm SO}} \sum_{k=1}^K p_k\Tr(\rho_kM_k) \\
    & \leq  \max_{{\cal M}} \sum_{k=1}^K p_k\Tr(\rho_kM_k)\\
    & = p_{\rm succ}({\cal E}),
\end{aligned}
\end{equation}
where the inequality naturally arises as a set of general POVMs includes a set of SO-implementable POVMs. For two equiprobable states, we have a simpler expression as in Eq.~\eqref{eq:bin_all},
\begin{equation} \label{eq:bin_disc}
    p_{\rm succ}^{\rm SO}({\cal E}^{\rm bin}) = \frac{1}{2} \left( 1 + \frac{\| \Delta \|_{\rm SO}}{2} \right),
\end{equation}
where 
\begin{equation} \label{eq:SO_norm}
    \frac{\|A\|_{\rm SO} }{2} := \max_{{\cal M}_{\rm SO}} \left|\Tr( A M)\right| \leq \max_{{\cal M}} \left|\Tr( A M)\right| = \frac{\| A \|_1 }{2} 
\end{equation}
denotes the SO-restricted trace norm of a Hermitian traceless operator $A$, defined similarly to the LOCC-restricted trace norm~\cite{matthews2009distinguishability}. 

\section{Limitations of stabilizer operations for discriminating stabilizer states}
\subsection{SO-indistinguishable stabilizer states}
We focus on a set of mutually orthogonal $n$-qubit stabilizer states and their discrimination using SOs. We state our main theorem as follows:
\begin{theorem}\label{thm:main}
There exists a set of mutually orthogonal stabilizer states ${\cal E}$ that cannot be perfectly discriminated by stabilizer operations, i.e.,
\begin{equation}
	p_{\rm succ}^{\rm SO}({\cal E}) < p_{\rm succ}({\cal E}) = 1.
\end{equation}
\end{theorem}
The key ideas behind Theorem~\ref{thm:main} are
\begin{enumerate}
\item[(i)] any stabilizer operations can be decomposed into a (possibly-adaptive) sequence of projectors,
\begin{equation}
\Pi_{\boldsymbol{a}} = \Pi_{a_L}^{P_L} \cdots  \Pi_{a_2}^{P_2} \Pi_{a_1}^{P_1},
\end{equation}
where $\Pi_{a}^{P} = \frac{\mathbb{1} + (-1)^{a} P}{2}$ are characterized by Pauli operators $P$ with binary outcomes $a \in \{ 0,1 \}$.
\item[(ii)] any informative measurement in the first round, $\Pi_{a_1}^{P_1}$, necessarily involves the disturbance of quantum states, making them no longer orthogonal.
\item[(iii)] adding a stabilizer-state ancilla does not help the discrimination, as any informative measurement in the joint system-ancilla states trivially on the ancilla; otherwise, the post-measurement state is perfectly recoverable.
\end{enumerate}
A complete version of the proof is given in Appendix~\ref{app:thm1}. To construct an explicit example in Theorem~\ref{thm:main}, we consider the following three-qubit stabilizer states,
\begin{equation}\label{eq:3q-states}
\begin{aligned}
	\ket{\psi_A^{\pm}} & = \frac{\ket{010} \pm \ket{110}}{\sqrt{2}} = \ket{\pm 10}\\
	\ket{\psi_B^{\pm}} & = \frac{\ket{001} \pm \ket{011}}{\sqrt{2}} = \ket{0{\pm}1}\\
	\ket{\psi_C^{\pm}} & = \frac{\ket{100} \pm \ket{101}}{\sqrt{2}} = \ket{1 0 \pm}\\
	\ket{\psi_{\rm GHZ}^{\pm}} & = \frac{\ket{000} \pm \ket{111}}{\sqrt{2}},
\end{aligned}
\end{equation}
where these states correspond to the edges and the space diagonal in the cube (see figure~\ref{fig:states}(a)).
\begin{figure*}[ht]
 \centering
        \includegraphics[width=.85\textwidth]{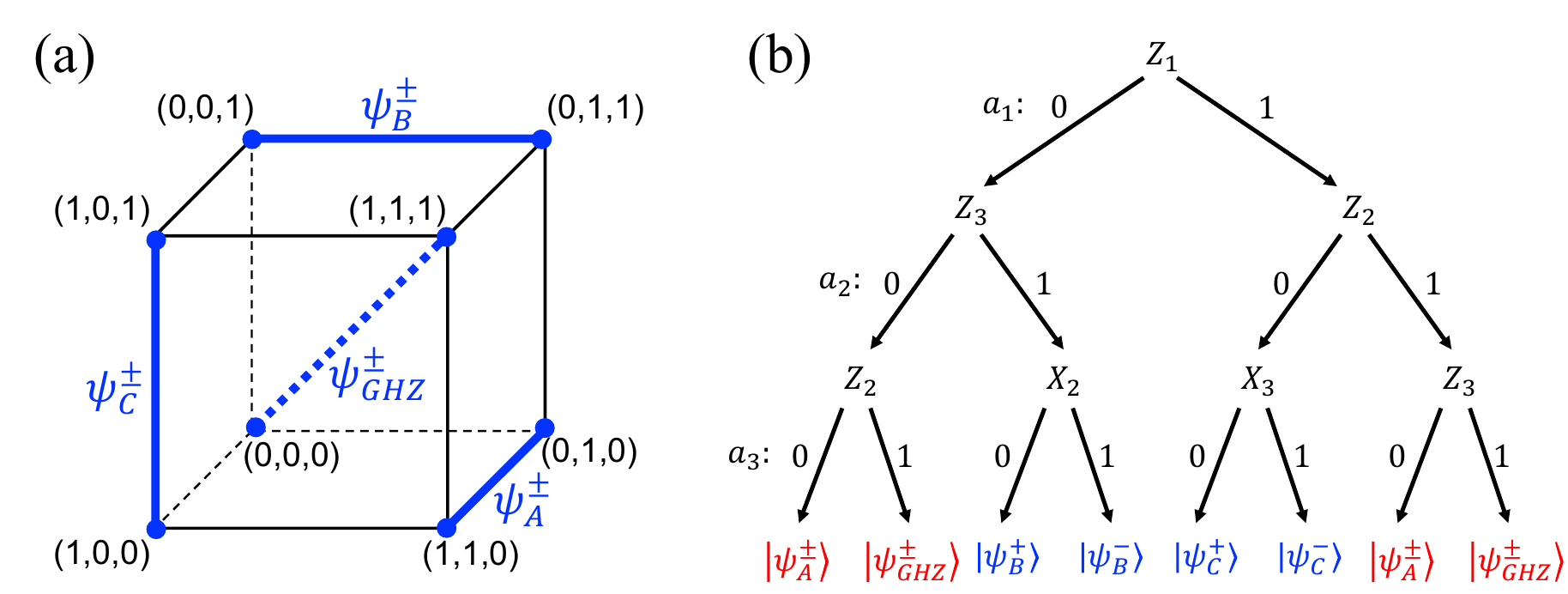}
 \caption{(a) Stabilizer basis of a $3$-qubit system. (b) Optimal Pauli-measurement decision-tree for discriminating the $3$-qubit stabilizer basis ${\cal E}_{\rm 3}$, which leads to $p_{\rm succ}^{\rm SO}({\cal E}_{\rm 3}) = 3/4$. $\ket{\psi_{A, {\rm GHZ}}^{\pm}}$ (highlighted in red) corresponds to the case where one cannot determine the sign $\pm$, so that the two states cannot be discriminated. The same strategy becomes optimal for ${\cal E}_{\rm 3, SEP}$ with $p_{\rm succ}^{\rm SO}({\cal E}_{\rm 3, SEP}) = 5/6$.}
\label{fig:states}
\end{figure*}

We note that these eight states are mutually orthogonal, thus forming a complete orthonormal basis. This implies that the set of these states (here, we denote the density matrix of the pure states as $\psi = \ket{\psi}\bra{\psi}$ for simplicity),
\begin{equation}
{\cal E}_{\rm 3} =\left\{ \psi_A^{\pm}, \psi_B^{\pm}, \psi_C^{\pm}, \psi_{\rm GHZ}^{\pm} \right\},
\end{equation}
can be perfectly discriminated by the corresponding projectors on the basis when any quantum measurement are allowed, i.e., $p_{\rm succ}({\cal E}_{\rm 3}) = 1$. On the other hand, when the allowed measurements are restricted to SO, achievable discrimination probability is limited to $p^{\rm SO}_{\rm succ}({\cal E}_{\rm 3}) = 3/4 < p_{\rm succ}({\cal E}_{\rm 3}) = 1$, which will be discussed in more detail in section~\ref{subsec:opt_succ}.

Even if we remove the two states $\ket{\psi_{\rm GHZ}^\pm}$ from the set to make every element product state,
\begin{equation}
	{\cal E}_{\rm 3, SEP} = \left\{ \psi_A^{\pm}, \psi_B^{\pm}, \psi_C^{\pm} \right\} ,
\end{equation}
cannot be perfectly discriminated by SO.  While such a set ${\cal E}_{\rm 3, SEP}$ was previously considered in the context of nonlocality without entanglement~\cite{Bennet99quantum}, their distinguishability by stabilizer operations has not been studied until now. When any one of these states is removed from ${\cal E}_{\rm 3, SEP}$, then they can be fully discriminated by SO (see Appendix~\ref{app:six_ness}).

A natural question is whether the six-state, three-qubit example, corresponding to $n=3$ and $K=6$, is minimal with respect to both the number of qubits $n$ and the number of distinct states $K$ among ensembles that cannot be perfectly discriminated by stabilizer operations. The following theorem establishes that at least three qubits and four states are necessary.
\begin{theorem}\label{thm:2qubit}
Any set of $K$ mutually orthogonal $n$-qubit pure stabilizer states can be perfectly discriminated by stabilizer operations whenever $n<3$ or $K<4$.
\end{theorem}
The proof is given in Appendix~\ref{app:thm2}. We further construct the smallest ensemble exhibiting imperfect distinguishability under stabilizer operations, consisting of $K=4$ states on $n=3$ qubits as follows:
\begin{equation}
	{\cal E}_{\rm 3, min} =\left\{ \psi_A^{+}, \psi_B^{+}, \psi_C^{+}, \psi_{-}^{\otimes 3} \right\},
\end{equation}
where $\psi_-^{\otimes 3} = (\ket{-}\bra{-})^{ \otimes 3}$ is also a product state. Such an ensemble has also been studied as an unextendible product basis of a three-qubit system~\cite{Bennett1999unextendible, bravyi2004unextendible, Duan2010locally}.

\begin{table*}[t]
\centering
\begin{tabular}{ l | c|  c| c | c | c |c }
\hline
$\E$ & $n$ & $K$ &  $p_{\rm succ}^{\rm SO}$ & Remarks  & $t_{p=1}$ \\
\hline \hline
${\cal E}_{\rm 3} =\left\{ \psi_A^{\pm}, \psi_B^{\pm}, \psi_C^{\pm}, \psi_{\rm GHZ}^{\pm} \right\} $ & 3 & 8 & 3/4  & Complete basis & 4\\
${\cal E}_{\rm 3, SEP} = \left\{ \psi_A^{\pm}, \psi_B^{\pm}, \psi_C^{\pm} \right\} $ & 3 & 6 & 5/6  & Product states & 4\\ 
${\cal E}_{\rm 3, min} =\left\{ \psi_A^{+}, \psi_B^{+}, \psi_C^{+}, \psi_{-}^{\otimes 3} \right\} $ & 3 & 4 & 15/16  & Minimal pure states & 3\\
${\cal E}_{\rm 2} =\left\{ \frac{\psi_{0} \otimes\psi_{+} + \psi_{+} \otimes \psi_{0} }{2}, \frac{\psi_{1}^{\otimes 2} + \psi_{-}^{\otimes 2} }{2} \right\}$ 
 & 2 & 2 & 7/8  & Minimal mixed states & $\infty$\\
\hline
\end{tabular}
\caption{Examples of a set of $n$-qubit stabilizer states with $n \leq 3$ that cannot be perfectly discriminated by SO. $\psi_{A,B,C,{\rm GHZ}}^{\pm}$ are described in Eq.~\eqref{eq:3q-states} and figure~\ref{fig:states}, and $\psi_{\alpha} = \ket{\alpha}\bra{\alpha}$ for $\alpha \in \{ 0, 1, +, -\}$.}
\label{tab:summary}
\end{table*}

On the other hand, when considering a set of orthogonal \textit{mixed} stabilizer states, one can find the following two-qubit example, even with only two states, i.e., $n=2$ and $K=2$, that cannot be perfectly discriminated by stabilizer operations:
\begin{equation}
	{\cal E}_{\rm 2} =\left\{ \frac{\psi_{0} \otimes\psi_{+} + \psi_{+} \otimes \psi_{0} }{2}, \frac{\psi_{1}^{\otimes 2} + \psi_{-}^{\otimes 2} }{2} \right\}.
\end{equation}

These sets of stabilizer states are summarized in table~\ref{tab:summary} along with the maximum success probabilities to discriminate the states by SO. Details regarding the optimal SO strategy will be discussed in section~\ref{subsec:opt_succ}.

\subsection{Optimal success probability using SO}
\label{subsec:opt_succ}

To construct an optimal SO strategy for discriminating a given ensemble of $n$-qubit states, we consider optimization over adaptive Pauli-decision trees. By the normal form of stabilizer operations~\cite{heimendahl2022axiomatic}, it suffices to consider adaptive Pauli-decision trees up to $n$ rounds. An adaptive Pauli-decision tree ${\cal T}$ assigns a Hermitian $n$-qubit Pauli operator, $P (\boldsymbol{a}_{<j})$ to each outcome history $\boldsymbol{a}_{<j}=(a_1,\ldots,a_{j-1})$ with $j \leq n$. The corresponding subsequent measurement at round $j$ with outcome history $\boldsymbol{a}_{<j}$ becomes $\Pi_{a_j}^{P(\boldsymbol{a}_{<j})}$ with binary outcome $a_j \in \{ 0,1\}$.

For a given outcome history $\boldsymbol{a}=(a_1,\ldots,a_r)$, so-called a \emph{leaf}, the corresponding projector becomes
\begin{equation}\label{eq:leaf_projector_main} 
	\Pi_{\boldsymbol{a}} := \prod_{j=1}^{r} \Pi_{a_j}^{P(\boldsymbol{a}_{<j})}. 
\end{equation}
For a set of leaves in the decision tree ${\cal T}$, denoted by ${\cal L}({\cal T})$, the Pauli observables along each root-to-leaf path can be chosen to be mutually commuting and independent. Moreover, since an $n$-qubit system admits at most $n$ independent commuting Pauli observables, the tree can be taken to have depth at most $n$. A formal definition is given in Appendix~\ref{app:Pauli_decision_tree}.

For discrimination of a given ensemble ${\cal E} = \left\{ (p_k,\rho_k) \right\}_{k=1}^{K}$, the optimal SO discrimination probability becomes
\begin{equation} \label{eq:SO_tree_optimization}
\begin{aligned}
p_{\mathrm{succ}}^{\mathrm{SO}}({\cal E}) &= \max_{{\cal T}, g } \sum_{\boldsymbol{a}\in{\cal L}({\cal T})} \sum_{k=1}^{K} p_k\, \delta_{k,g(\boldsymbol{a})} \operatorname{Tr} \left( \Pi_{\boldsymbol{a}}\rho_k \right)\\
& = \max_{{\cal T}} \sum_{\boldsymbol{a}\in{\cal L}({\cal T})} \max_{ k} p_k \operatorname{Tr} \left( \Pi_{\boldsymbol{a}}\rho_k \right),
\end{aligned}
\end{equation}
where $g(\boldsymbol{a}) \in \{ 1, 2, \cdots, K\}$ is a guessing function from the outcomes $\boldsymbol{a}$ to the state label $k$. The second equality follows because the announced label can be optimized independently at each leaf. We also note that since classical randomization of multiple decision trees cannot improve the success probability, it suffices to optimize over deterministic trees and decision rules.

As an illustrative example, the optimal Pauli-decision tree for the complete three-qubit stabilizer basis ${\cal E}_{3}$ is shown in figure~\ref{fig:states}(b). The optimal Pauli-decision tree perfectly identifies the four states $\psi_B^{\pm}$ and $\psi_C^{\pm}$,while distinguishing the $A$ and ${\rm GHZ}$ branches without resolving their signs. Hence, each of $\psi_A^{\pm}$ and $\psi_{\rm GHZ}^{\pm}$ is identified with probability $1/2$. For the uniform ensemble, this gives $ p_{\mathrm{succ}}^{\mathrm{SO}} \left({\cal E}_{3}\right) = \frac{1}{8} \left( 4+4\times\frac{1}{2} \right) = \frac{3}{4}$.  The optimal success probabilities for the remaining ensembles are given in table~\ref{tab:summary}.

For a general ensemble, however, enumerating all adaptive trees remains computationally expensive. At each node at depth $r$, the preceding measurements could be chosen under $r$ independent commuting stabilizer constraints, leaving $4^{n-r}-1$ inequivalent nontrivial Pauli refinements. Since a complete binary tree has $2^r$ nodes at depth $r$, the number of possible adaptive Pauli-decision trees scales as $\prod_{r=0}^{n-1} \left(4^{n-r}-1\right)^{2^r} = 2^{\Theta(2^n)}$, which is doubly exponential in $n$.

Nevertheless, the optimization cost can be significantly reduced by adopting dynamic programming~\cite{Bellman1952}. To this end, we take the stabilizer-code projector associated with each branch as the state variable instead of explicitly enumerating all decision trees. Since this projector specifies the subspace surviving the preceding measurements, the optimal continuation depends only on the current projector rather than on the full measurement history. This replaces the doubly exponential enumeration of $2^{\Theta(2^n)}$ complete tree labelings with an optimization over at most $2^{\Theta(n^2)}$ distinct stabilizer projectors. Further details are provided in Appendix~\ref{app:dynamic_programming}.

\subsection{Multi-qubit generalization}
We show that the construction of $3$-qubit stabilizer states in Eq.~\eqref{eq:3q-states} can be generalized to any $n$-qubit states with $n\geq3$ by considering the following set of quantum states with three $n_a$, $n_b$, and $n_c$-qubit subsystems with $n = n_a + n_b + n_c$:
\begin{equation}\label{eq:n-qubit_states}
\begin{aligned}
	\ket{\psi_A(\boldsymbol{x}_a, \boldsymbol{z}_b^1, \boldsymbol{z}_c^0)} &= \ket{\boldsymbol{x}_a} \otimes \ket{\boldsymbol{z}_b^{1}} \otimes \ket{\boldsymbol{z}_c^{0}} \\
	\ket{\psi_B(\boldsymbol{z}_a^0, \boldsymbol{x}_b, \boldsymbol{z}_c^1)} &= \ket{\boldsymbol{z}_a^0} \otimes \ket{\boldsymbol{x}_b} \otimes \ket{\boldsymbol{z}_c^{1}} \\
	\ket{\psi_C(\boldsymbol{z}_a^1, \boldsymbol{z}_b^0, \boldsymbol{x}_c)} &= \ket{\boldsymbol{z}_a^1} \otimes \ket{\boldsymbol{z}_b^{0}} \otimes \ket{\boldsymbol{x}_c} \\
	\ket{\psi_D(\boldsymbol{z}_a^0, \boldsymbol{z}_b^0, \boldsymbol{z}_c^0 )} &= \ket{\boldsymbol{z}_a^0} \otimes \ket{\boldsymbol{z}_b^{0}} \otimes \ket{\boldsymbol{z}_c^0} \\
	\ket{\psi_E(\boldsymbol{z}_a^1, \boldsymbol{z}_b^1, \boldsymbol{z}_c^1)} &= \ket{\boldsymbol{z}_a^1} \otimes \ket{\boldsymbol{z}_b^{1}} \otimes \ket{\boldsymbol{z}_c^1}.
	\end{aligned}
\end{equation}
Here, for each $\mu = a,b,c$, we define $\boldsymbol{x}_\mu \in \{ +, -\}^{ n_\mu}$ and $\boldsymbol{z}_\mu^{0,1} \in {\cal Z}_\mu^{0,1} \subset \{ 0, 1\}^{n_\mu}$ are all possible bit-strings such that $f_\mu (\boldsymbol{z}_\mu^0) = 0$ and $f_\mu (\boldsymbol{z}_\mu^1) = 1$ for a Boolean function $f_\mu(\boldsymbol{z}_\mu) : \{0,1\}^{n_\mu} \to \{0,1\}$. As a Boolean function cannot have both $0$ and $1$ for the same input, we have ${\cal Z}_\mu^0 \cap {\cal Z}_\mu^1 = \emptyset$ and ${\cal Z}_\mu^0 \cup {\cal Z}_\mu^1 = \{ 0, 1\}^{n_\mu}$. This makes the states in Eq.~\eqref{eq:n-qubit_states} mutually orthogonal and a complete basis of $n$-qubit systems. When $n_a=n_b=n_c=1$ and $f_a(z)=f_b(z)=f_c(z)=z$, these states reduce to the three-qubit basis ${\cal E}_{\rm 3}$.

The key idea for demonstrating SO-indistinguishability is to take $f_\mu$ for all $\mu = a, b, c$ with vanishing linear structure~\cite{carlet2021boolean}, which satisfies the property that, for every $\boldsymbol\gamma_\mu \neq \boldsymbol{0}$, there exists $\boldsymbol z_\mu$ such that $f_\mu(\boldsymbol{z}_\mu)\neq f_\mu (\boldsymbol{z}_\mu \oplus \boldsymbol{\gamma}_\mu)$. A typical example of such a function for even $n_\mu$ is a bent function~\cite{ROTHAUS1976300}. The vanishing linear structure guarantees that for every Pauli observable $P \neq \mathbb{1}$, there exist distinct states
$\ket{\psi_k}$ and $\ket{\psi_{k'}}$ in the set
\begin{widetext}
\begin{equation}
	{\cal E}_{n, {\rm SEP}} = \{ \psi_A(\boldsymbol{x}_a, \boldsymbol{z}_b^1, \boldsymbol{z}_c^0), \psi_B(\boldsymbol{z}_a^0, \boldsymbol{x}_b, \boldsymbol{z}_c^1), \psi_C(\boldsymbol{z}_a^1, \boldsymbol{z}_b^0, \boldsymbol{x}_c) \}_{\boldsymbol{x}_{a,b,c} \in \{ +, -\}^{ n_{a,b,c}}, \boldsymbol{z}_{a,b,c}^{0,1} \in {\cal Z}_{a,b,c}^{0,1}},
\end{equation}
\end{widetext}
such that  $\bra{\psi_k} P\ket{\psi_k}=0$ and $\bra{ \psi_{k'}} P \ket{ \psi_k} \neq 0$. Hence, the first nontrivial Pauli measurement in any SO decision tree makes some pair of candidate states nonorthogonal in each measurement branch (see Appendix~\ref{app:n-qubit} for details). As no subsequent operation can perfectly distinguish that pair, we conclude that
\begin{equation} \label{eq:multiqubit_SO_impossibility}
    p_{\mathrm{succ}}^{\mathrm{SO}} ({\cal E}_{n, {\rm SEP}})<1.
\end{equation}
One can further add stabilizer states described by a linear combination of $\ket{\psi_D(\boldsymbol{z}_a^0, \boldsymbol{z}_b^0, \boldsymbol{z}_c^0 )}$ and $\ket{\psi_E(\boldsymbol{z}_a^1, \boldsymbol{z}_b^1, \boldsymbol{z}_c^1)}$, for example, $\ket{\psi_{DE}^{\pm}} = \frac{1}{\sqrt{2}} \left( \ket{\psi_D( \boldsymbol{z}_a^0, \boldsymbol{z}_b^0, \boldsymbol{z}_c^0)} \pm \ket{\psi_E( \boldsymbol{z}_a^1, \boldsymbol{z}_b^1, \boldsymbol{z}_c^1)} \right)$ for a fixed $\boldsymbol{z}_{a,b,c}^{0,1}$, to make the discrimination task more complicated as in ${\cal E}_{\rm 3}$.

\subsection{Perfect discrimination using non-stabilizer operations}
So far, we have shown that there exists a set of mutually orthogonal stabilizer states that cannot be perfectly discriminated by SO. When we allow quantum operations beyond SO, these states could always be perfectly discriminated. For example, states in ${\cal E}_{\rm 3,SEP}$ can be perfectly discriminated by a non-Clifford unitary circuit described in figure~\ref{fig:nonstab_circuit}, as a sequence of controlled-controlled-Hadamard (CCH) gates converts each state into an independent computational-basis state.
\begin{figure}[ht]
\hspace{.0\linewidth}
\begin{minipage}[t]{.5\linewidth}
\vspace{-32pt}
$ 
\Qcircuit @C=.5em @R=1.2em @!R{
&&{\gategroup{1}{3}{3}{3}{1.em}{\{}} & \ctrl{1} & \qw & \ctrlo{1} & \qw & \gate{H} & \qw & \meter \\
&\lstick{\ket{\psi_k}} && \ctrlo{1} & \qw & \gate{H} & \qw & \ctrl{-1} & \qw & \meter \\
&&& \gate{H} & \qw & \ctrl{-1} & \qw & \ctrlo{-1} & \qw & \meter }
$
\end{minipage}
\begin{tabular}{ | c | c  |}
\hline 
Input  & Output \\
\hline 
$| \psi_A^{+} \rangle = \ket{+}\ket{1}\ket{0}$ & $\ket{0}\ket{1}\ket{0}$ \\
$| \psi_B^{+} \rangle = \ket{0}\ket{+}\ket{1}$ & $\ket{0}\ket{0}\ket{1}$ \\
$| \psi_C^{+} \rangle = \ket{1}\ket{0}\ket{+}$ & $\ket{1}\ket{0}\ket{0}$ \\
$| \psi_A^{-} \rangle = \ket{-}\ket{1}\ket{0}$ & $\ket{1}\ket{1}\ket{0}$ \\
$| \psi_B^{-} \rangle = \ket{0}\ket{-}\ket{1}$ & $\ket{0}\ket{1}\ket{1}$ \\
$| \psi_C^{-} \rangle = \ket{1}\ket{0}\ket{-}$ & $\ket{1}\ket{0}\ket{1}$ \\
\hline
\end{tabular}
\caption{Quantum circuit that perfectly distinguishes  ${\cal E}_{\rm 3,SEP}$}
\label{fig:nonstab_circuit}
\end{figure}

\noindent This also demonstrates that, although each individual stabilizer state can be mapped to a computational-basis state by a suitable Clifford unitary, there need not exist a single Clifford unitary that simultaneously maps a given orthogonal set of stabilizer states to distinct computational-basis states.

Furthermore, when allowing a general adaptive strategy, using a single Toffoli gate enables perfect discrimination of both ${\cal E}_{\rm 3, SEP}$ and ${\cal E}_{\rm 3}$ (see Appendix~\ref{app:3q_protocol} for the explicit protocol). As a single Toffoli gate is less costly than implementing three CCH gates from a $T$-gate or $T$-states, this result implies that an adaptive protocol could be advantageous. A more detailed scenario with an adaptive strategy assisted by a non-stabilizer ancilla state, e.g., $T$-state injection, will be discussed in section~\ref{subsec:magic_cost}.

\section{Testing nonstabilizerness using stabilizer states}
\subsection{Data hiding in stabilizer states}
\label{subsec:data_hiding}

Quantum data hiding refers to encoding classical information into globally distinguishable quantum states whose distinguishability is strongly reduced when the allowed measurements are restricted~\cite{divincenzo2002quantum, matthews2009distinguishability}. While conventional quantum data-hiding protocols commonly impose locality restrictions such as LOCC, here we instead restrict the decoder to stabilizer operations (SO). This leads to a notion of data hiding in which both encoding states are mixtures of stabilizer states, but recovering the encoded information perfectly requires operations outside SO.

To construct a data-hiding scheme in stabilizer states, we explicitly construct a binary ensemble ${\cal E}^{\rm bin}(\rho_0,\rho_1) = \left\{ \rho_0 , \rho_1\right \}$ where $\rho_0$ and $\rho_1$ are orthogonal stabilizer states that are not perfectly discriminated by SO. Based on the three-qubit complete stabilizer basis ${\cal E}_{\rm 3}$, one can encode a classical bit into the mixture of the basis as ${\cal E}_3^{\rm bin} = \{ \rho_0, \rho_1\}$, where
\begin{equation}
     \rho_0 = \frac{1}{4} \sum_{\alpha = A, B, C, {\rm GHZ}} \psi^{+}_\alpha, \quad \rho_1 = \frac{1}{4} \sum_{\alpha = A, B, C, {\rm GHZ}} \psi^{-}_\alpha,
\end{equation}
by encoding the bit in the phase. Since the two density operators $\rho_0$ and $\rho_1$ have orthogonal supports,  when allowing all possible quantum operations, perfect discrimination of ${\cal E}_3^{\rm bin}$ is possible, i.e., $p_{\rm succ} ({\cal E}_3^{\rm bin}) =1$. On the other hand, the optimal SO from the construction in section~\ref{subsec:opt_succ} leads to the limited success probability of
\begin{equation}
    p_{\rm succ}^{\rm SO} ({\cal E}_3^{\rm bin}) = \frac{3}{4},
\end{equation}
so that SO cannot recover the hidden bit perfectly. Remarkably, the optimal SO-discrimination probability of these coarse-grained states in ${\cal E}_3^{\rm bin}$ is the same as that in the complete basis set ${\cal E}_3$, i.e., $p_{\rm succ}^{\rm SO} ({\cal E}_3^{\rm bin}) = p_{\rm succ}^{\rm SO} ({\cal E}_3) = 3/4$. This implies that revealing the phase of these states is as difficult as fully discriminating all $8$ basis states, when only using SOs.

One can similarly construct another binary ensemble from ${\cal E}_{\rm 3, SEP}$ as
\begin{equation}
	{\cal E}_{\rm 3, SEP}^{\rm bin} = \left\{ \frac{1}{3} \sum_{\alpha = A, B,C} \psi_\alpha^{+}, \frac{1}{3} \sum_{\alpha = A, B,C} \psi_\alpha^{-}\right\},
\end{equation}
where this ensemble also has the same SO-discrimination probability as the original ensemble, i.e., $p_{\rm succ}^{\rm SO} ({\cal E}_{\rm 3, SEP}^{\rm bin}) = p_{\rm succ}^{\rm SO} ({\cal E}_{\rm 3, SEP}) = 5/6$.  As the two-qubit mixed ensemble ${\cal E}_{\rm 2}$ already contains two states, it can be directly utilized for data hiding.

The performance of data hiding could be quantified by the data-hiding ratio~\cite{matthews2009distinguishability}, defined as 
\begin{equation}
    R_{\rm hiding}({\cal E}^{\rm bin}) := \frac{ p_{\rm succ}({\cal E}^{\rm bin})-\frac{1}{2}}{ p_{\rm succ}^{\rm SO}({\cal E}^{\rm bin})-\frac{1}{2}}
    = \frac{\|\Delta\|_1}{\|\Delta\|_{\rm SO}},
    \label{eq:data_hiding_ratio}
\end{equation}
where the trace norms with unrestricted and SO-restricted operations $\| A \|_1$ and $\| A \|_{\rm SO}$ are defined in Eq.~\eqref{eq:SO_norm}. The data hiding ratios of the three-qubit ensembles, ${\cal E}_3^{\rm bin}$ and ${\cal E}_{\rm 3, SEP}^{\rm bin}$,  are given as
\begin{equation}
    R_{\rm hiding}({\cal E}_3^{\rm bin}) = 2, \qquad R_{\rm hiding}({\cal E}_{\rm 3, SEP}^{\rm bin}) = \frac{3}{2},
\end{equation}
as $\| \Delta \|_1 = 2$ for mutually orthogonal states $\rho_0$ and $\rho_1$. Table~\ref{tab:binary} summarizes the binary ensembles of stabilizer states and their data hiding ratios.

\begin{figure*}[t]
 \centering
        \includegraphics[width=.9\textwidth]{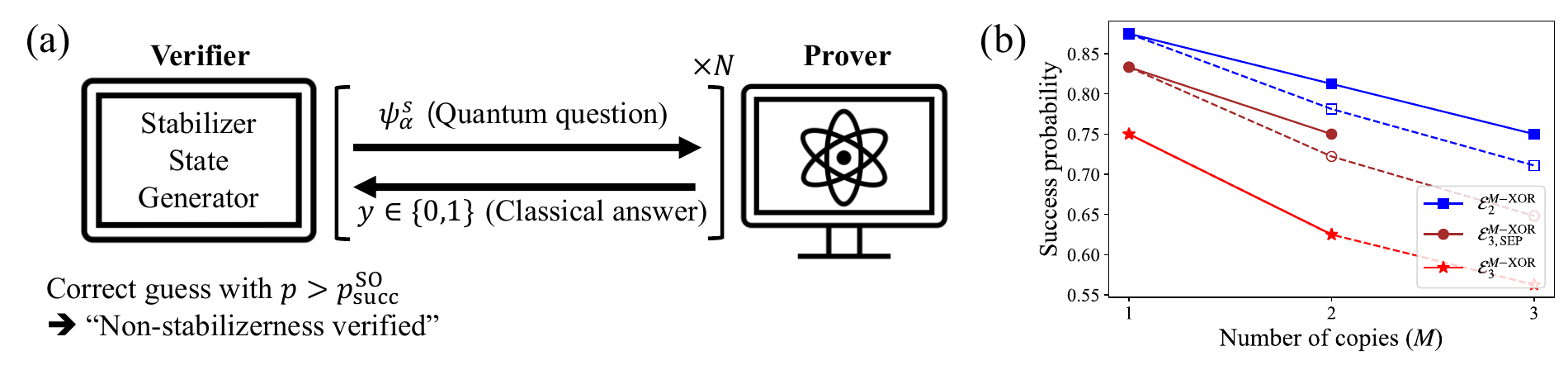}
 \caption{(a) Verifier-prover protocol for testing nonstabilizerness (b) Success probability of $M$-XOR guessing task for $M=1,2,3$.}
\label{fig:test}
\end{figure*}

\begin{table}[!ht]
\centering
\begin{tabular}{ l | c| c| c | c | c |c }
\hline
$\E$ & $n$ &  $p_{\rm succ}^{\rm SO}$ & $R_{\rm hiding}$ & Collective advantage & $t_{p=1}$ \\
\hline \hline
${\cal E}_{\rm 3}^{\rm bin}$ & 3 & 3/4  & $2$ & No for $M=2$ & 4\\
${\cal E}_{\rm 3, SEP}^{\rm bin} $ & 3 & 5/6  & $3/2$ & Yes & 3\\ 
${\cal E}_{\rm 2}^{\rm bin}  = {\cal E}_2$  & 2 & 7/8  & $4/3$ & Yes & $\infty$\\
\hline
\end{tabular}
\label{tab:binary}
\caption{Examples of binary stabilizer states for data hiding against SO.}
\end{table}

Compared with the discrimination of magic states and their orthogonal complements in odd-dimensional quantum systems~\cite{Zu24limitations}, our construction employs only qubit stabilizer states. The encoded states can therefore be prepared as logical stabilizer states and protected using standard fault-tolerant stabilizer error-correction procedures. Perfect decoding, by contrast, requires operations outside SO and therefore consumes precisely the resources that are costly in conventional fault-tolerant architectures.

\subsection{Testing nonstabilizerness in a verifier--prover protocol}
\label{subsec:verification}
Based on the data-hiding property of the binary stabilizer-state ensemble, we construct a verifier--prover protocol for testing whether a device can implement operations beyond stabilizer operations. In this protocol, the verifier sends a challenge to the prover, who then returns an answer. Notably, the verifier needs only to prepare stabilizer states and therefore does not require access to any non-stabilizer operations. Moreover, only the classical information returned by the prover is required to test nonstabilizerness (see figure~\ref{fig:test}). An explicit protocol based on the complete three-qubit ensemble ${\cal E}_3$ and its binary counterpart ${\cal E}_3^{\rm bin}$ proceeds as follows:
\begin{enumerate}
    \item In each round, the verifier randomly prepares a pure stabilizer state $\psi_\alpha^{s}$ with $s \in \{ +, - \}$ and $\alpha \in \{ A, B, C, {\rm GHZ}\}$ from the ensemble ${\cal E}_3$ and sends it to the prover.
    \item The prover returns a bit $y\in\{0,1\}$ as its guess for the phase label $s$ of the stabilizer state.
    \item The verifier checks whether the guess is correct, i.e., $(-1)^y=s$.
    \item The verifier repeats the protocol $N$ times and records the number of successful guesses, $N_{\rm succ}$.
    \item If the observed success rate $p_{\rm succ}^{\rm exp}:=\frac{N_{\rm succ}}{N}$ exceeds the SO bound $p_{\rm succ}^{\rm SO}({\cal E}_3^{\rm bin})$, the result violates the SO benchmark. Equivalently, the data revealing ratio satisfies
    \begin{equation}
        R_{\rm reveal} := \frac{p_{\rm succ}^{\rm exp}-1/2} {p_{\rm succ}^{\rm SO}({\cal E}_3^{\rm bin})-1/2} >1.
    \end{equation}
\end{enumerate}

We emphasize that, although the verifier prepares pure states from ${\cal E}_3$, the prover is asked to guess only their phase labels, corresponding to discrimination of ${\cal E}_3^{\rm bin}$. Nevertheless, this provides the same certification threshold because $p_{\rm succ}^{\rm SO}({\cal E}_3^{\rm bin}) = p_{\rm succ}^{\rm SO}({\cal E}_3)$, as shown in section~\ref{subsec:data_hiding}. The same relation holds between
${\cal E}_{\rm 3,SEP}$ and ${\cal E}_{\rm 3,SEP}^{\rm bin}$, in which case nonstabilizerness can be verified using only separable stabilizer-state preparations.

When the challenges are presented sequentially, and each answer must be returned before the next state is sent, the prover's conditional probability of answering correctly in each round is at most $q_{\rm SO} = p_{\rm succ}^{\rm SO}$ under the SO null hypothesis. Consequently, the total number of successes is stochastically dominated by a binomial random variable with success probability $q_{\rm SO} $. For an observed value $N_{\rm succ}$, a valid one-sided $p$-value under the SO null hypothesis is therefore
\begin{equation}
    p_{\rm val} = \sum_{j=N_{\rm succ}}^{N} \binom{N}{j} q_{\rm SO} ^{j} (1-q_{\rm SO})^{N-j}.
    \label{eq:verification_pvalue}
\end{equation}
The SO null hypothesis can thus be rejected at any significance level greater than or equal to $p_{\rm val}$. For example, for $N=1000$ and $N_{\rm succ}=900$, with $q_{\rm SO} = p_{\rm succ}^{\rm SO}({\cal E}_3^{\rm bin})=3/4$, we obtain $p_{\rm val} \simeq 2.13\times10^{-33}$, corresponding to a one-sided Gaussian significance of approximately $11.98\sigma$. In practice, some errors, for example, in state preparation and noise throughout the transmission, can be incorporated conservatively by replacing $q_{\rm SO}$ with $q_{\rm SO}+\epsilon$, where $\epsilon$ is an upper bound on the average trace distance between the intended and prepared challenge states.

The main advantage of this protocol is that it provides information-theoretic certification of an operation outside SO without relying on complexity-theoretic or cryptographic assumptions. In this respect, it differs from protocols designed to verify the correctness of arbitrary delegated quantum computations~\cite{gheorghiu2019verification}. More practically, this nonstabilizerness test does not require quantum-state or process tomography, which is not only costly in terms of measurement shots but also relies entirely on the prover's description of the device. The limitation of the present protocol is that it is only a capability test: it certifies the presence of a nonstabilizer operation or resource in the prover's implementation but does not identify a particular non-Clifford gate or establish the correctness of a general quantum computation. Nevertheless, in section~\ref{subsec:magic_cost}, we show that the amount of nonstabilizerness can also be quantitatively bounded under the assumption that the device implements quantum circuits via magic-state injection~\cite{bravyi2005universal,campbell2017roads}. This makes our protocol particularly well suited to the early fault-tolerant era, as the verifier needs only to prepare stabilizer states, while the prover must demonstrate the nonstabilizer capability enabled by magic states, whose fault-tolerant preparation remains resource intensive~\cite{chamberland2020magic}.

\subsection{$M$-XOR guessing task} \label{subsec:M-XOR}
The nonstabilizer verification protocol could be further improved by amplifying the data hiding ratio via encoding one bit into the parity of multiple independently prepared blocks. To this end, we construct an $M$-XOR guessing task for a positive integer $M$, by defining
\begin{equation}
\begin{aligned}
&    {\cal E}^{M\textrm{-XOR}}(\rho_0,\rho_1) := \left\{ \rho_0^{(M)}, \rho_1^{(M)} \right\},\\
    & \rho_r^{(M)} = \frac{1}{2^{M-1}} \sum_{\substack{r_1,\ldots,r_M\in\{0,1\}\\ r_1\oplus\cdots\oplus r_M=r }}
    \bigotimes_{m=1}^{M}\rho_{r_m},
\end{aligned}
    \label{eq:M_XOR_ensemble}
\end{equation}
where a similar construction has recently been considered in data hiding against LOCC~\cite{Ha2025quantum, mele2025optimising}. We note that the binary ensemble ${\cal E}^{\rm bin} = \{\rho_0 , \rho_1\}$ is the special case with $M=1$. Furthermore, as both $\rho_0$ and $\rho_1$ are stabilizer states, increasing $M$ makes the discrimination task more difficult as
\begin{equation}
	p_{\rm succ}^{\rm SO} \left( {\cal E}^{\rm bin} \right) \geq p_{\rm succ}^{\rm SO} \left( {\cal E}^{M'\textrm{-XOR}} \right) \geq p_{\rm succ}^{\rm SO} \left( {\cal E}^{M\textrm{-XOR}} \right),
\end{equation}
for all $ M' \leq M$. To see this, one can perform the discrimination task between $\rho^{(M')}_r$ with $r \in \{0,1\}$ by appending randomly chosen $\rho_{r_{M-M'}}^{(M - M')}$ with known index $r_{M-M'}$, and then perform the optimal $M$-XOR guessing protocol for the appended state $\rho^{(M')}_r \otimes \rho_{r_{M-M'}}^{(M - M')}$.  As the preparation of these additional states is allowed within SO, the overall protocol belongs to SO. The index of the target state $r$ can then be inferred directly from the total parity $r_M$ and the known indices $r_{M-M'}$ of the appended states, by a simple relation $r = r_M \oplus r_{M-M'}$. As a direct consequence, any $M$-XOR guessing is at least as difficult as the binary discrimination task.

A straightforward SO strategy for the $M$-XOR guessing task is to guess each $r_m$ independently and then compute $r=\bigoplus_{m=1}^M r_m$. This strategy succeeds whenever an even number of the individual guesses are incorrect. Its success probability is therefore
\begin{equation}
\begin{aligned}
&p_{\rm succ}^{\rm SO,ind.} \left({\cal E}^{M\textrm{-XOR}}\right) \\
&\quad = \sum_{\substack{k=0~k:{\rm even}}}^{M} \binom{M}{k} \left(p_{\rm succ}^{\rm SO}\right)^{M-k} \left(1-p_{\rm succ}^{\rm SO}\right)^k \\
&\quad = \frac{1}{2} \left[ 1+\left(2p_{\rm succ}^{\rm SO}-1\right)^M \right],
\end{aligned}
\end{equation}
where $p_{\rm succ}^{\rm SO}:=p_{\rm succ}^{\rm SO}({\cal E}^{\rm bin})$.

Interestingly, for the three-qubit binary phase ensemble ${\cal E}_3^{\rm bin}$ induced by the complete basis ${\cal E}_3$, an exhaustive search over all SO protocols yields
\begin{equation}
p_{\rm succ}^{\rm SO}\left({\cal E}_3^{2\textrm{-XOR}}\right) = p_{\rm succ}^{\rm SO, ind.}\left({\cal E}_3^{2\textrm{-XOR}}\right)= \frac{5}{8},
\end{equation}
showing that there is no collective advantage for $M=2$. Motivated by this observation, we conjecture that no collective advantage exists for any $M$.

\begin{conjecture}\label{conj:M-XOR_3q}
For any positive integer $M$, the $M$-XOR ensemble induced by ${\cal E}_3^{\rm bin}$ has no collective advantage, i.e., 
\begin{equation} \label{eq:conj_prob_3q}
p_{\rm succ}^{\rm SO}\left({\cal E}_3^{M\textrm{-XOR}}\right) = p_{\rm succ}^{\rm SO, ind.}\left({\cal E}_3^{M\textrm{-XOR}}\right) = \frac{1}{2}\left(1+2^{-M}\right),
\end{equation}
or equivalently,
\begin{equation}
	R_{\rm hiding}({\cal E}_3^{M\textrm{-XOR}}) = 2^M.
\end{equation}
\end{conjecture}
If this conjecture holds, the success probability approaches the random-guessing probability $1/2$ exponentially fast as $M$ increases. On the other hand, we find collective advantages for the other binary ensembles, ${\cal E}_{\rm 3,SEP}$ and ${\cal E}_2$, as illustrated in figure~\ref{fig:test}(b). Even in these cases, however, the success probability decreases significantly as $M$ increases, leaving open the possibility that it also decays exponentially toward $1/2$.

The proposed $M$-XOR guessing task can be utilized for  nonstabilizerness verification, where the verifier now prepares a product of pure stabilizer states
\begin{equation}
	\Psi_{\boldsymbol{\alpha}}^{\boldsymbol {s}} = \psi_{\alpha_1}^{s_1} \otimes \psi_{\alpha_2}^{s_2} \otimes \cdots \otimes \psi_{\alpha_M}^{s_M},
\end{equation}
where each $\psi_{\alpha_m}^{s_m}$ with $s_m \in \{ +, -\}$ and $\alpha_m \in \{ A, B, C, {\rm GHZ} \}$ are randomly chosen from ${\cal E}_3$. The prover is then required to return the total parity of the sign as $(-1)^y = s_1 s_2 \cdots s_M$. Notably, if the conjectured absence of a collective advantage or weaker exponential decay of success probability toward $1/2$ holds, increasing $M$ exponentially suppresses the SO bias above random guessing, whereas an unrestricted prover can still determine the parity perfectly.

\begin{figure*}[!t]
    \centering
    \includegraphics[width=.9\textwidth]{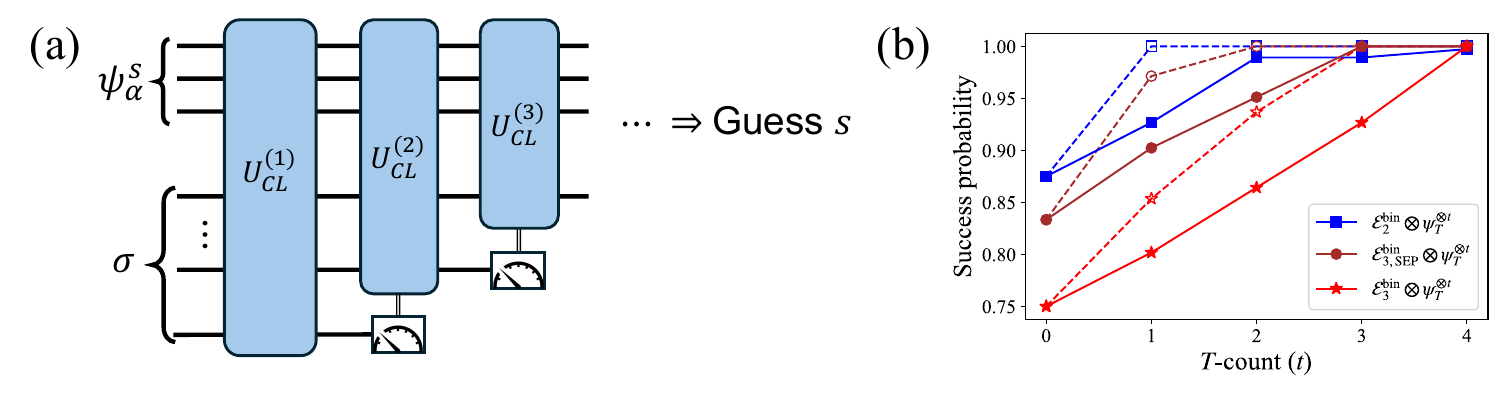}
    \caption{(a) Adaptive SO discrimination using nonstabilizer-state injected Clifford circuits and (b) Numerically obtained success probability (solid lines) and the robustness bound (dashed lines) with $t$ copies of injected $T$-states.}
    \label{fig:magic_bound}
\end{figure*}

\subsection{Nonstabilizer cost for perfect discrimination and robustness bound}
\label{subsec:magic_cost}

The proposed test not only verifies the presence of nonstabilizerness in a quantum computing device, but also provides a quantitative lower bound on its amount when considering Clifford circuits supplemented by magic-state injection, which constitute a standard architecture for universal fault-tolerant quantum computation~\cite{bravyi2005universal}. In this architecture, all gates, measurements, and classically controlled corrections are SOs, which are fault-tolerantly implemented in stabilizer codes, while nonstabilizerness is supplied through ancillary nonstabilizer states prepared separately using magic-state distillation~\cite{bravyi2005universal} or cultivation~\cite{gidney2024cultivation} protocols.

Let us suppose that the prover is supplied with an ancillary state $\sigma$, which may be a nonstabilizer state while keeping the remaining operations restricted to SO, as illustrated in figure~\ref{fig:magic_bound}(a). An important example is the $T$-state
\begin{equation}
\psi_T := \ket{T} \bra{T}, \qquad \ket{T} = \frac{\ket{0}+e^{i\pi/4}\ket{1}}{\sqrt{2}},
\label{eq:T_state}
\end{equation}
where the non-Clifford gates, such as the $T$-gate~\cite{zhou2000methodology} and the Toffoli gate~\cite{Jones2023low}, can be implemented using SO by consuming a single copy and four copies of $\psi_T$, respectively. Since the addition of these non-Clifford gates allows for the construction of a universal gate set~\cite{nielsen2010quantum}, this architecture supports universal quantum computation given a sufficient amount of nonstabilizerness in the ancillary state $\sigma$.

In the discrimination problem of our focus, the optimal $\sigma$-assisted SO discrimination probability for a given stabilizer state ensemble ${\cal E}$ now becomes 
\begin{equation}
	p_{\rm succ}^{\rm SO} ({\cal E}; \sigma) = p_{\rm succ}^{\rm SO} ({\cal E} \otimes \sigma) = \max_{{\cal M}_{\rm SO}} \sum_{k=1}^K p_k\Tr\left[ (\rho_k \otimes \sigma) M_k\right],
\end{equation}
where ${\cal E}\otimes\sigma := \left\{ \rho_k\otimes\sigma \right\}_{k=1}^{K}$ is an ensemble appended by $\sigma$. As the optimization over ${\cal M}_{\rm SO}$ remains SO-implementable measurements in a larger dimension, the Pauli-decision tree protocol introduced in section~\ref{subsec:opt_succ} is also applicable to evaluate the optimal success probability. For example, one can explicitly evaluate the required number of $T$-states for perfect discrimination by solving the optimization problem,
\begin{equation}
	t_{p=1}({\cal E}) := \min \{ t: p_{\rm succ}^{\rm SO} ({\cal E}; \psi_T^{\otimes t}) = 1, t \in \mathbb{Z}^+ \}.
\end{equation}
The optimal $t$-counts for the stabilizer ensembles are given in tables~\ref{tab:summary} and \ref{tab:binary}. An interesting feature of ${\cal E}_2$ is that perfect discrimination cannot be achieved using any finite number of $T$-states, although the success probability can be made arbitrarily close to unity by increasing the number of  $T$-states (see Appendix~\ref{app:finite_T_E2} for details).

For a general non-stabilizer ancillary state $\sigma$, the amount of nonstabilizerness can be quantified by the robustness of magic~\cite{Howard17application}.
\begin{definition} [Robustness of magic~\cite{Howard17application}] For a given $n$-qubit state $\sigma$, the robustness of magic is defined as
\begin{equation}
\begin{aligned}
    &{\cal R}(\sigma) \\
    &:= \min_{ \{ q_i, \ket{\phi_i} \} } \left\{ \sum_i |q_i| : \sigma = \sum_i q_i \ket{\phi_i}\bra{\phi_i}, \forall \ket{\phi_i} \in \mathrm{STAB}_n \right\},
\end{aligned}
    \label{eq:robustness_magic}
\end{equation}
where $q_i$ are real coefficients with $\sum_i q_i = 1$ and $\mathrm{STAB}_n$ is a set of $n$-qubit stabilizer states.
\end{definition}
We note that ${\cal R}(\sigma)=1$ if and only if $\sigma$ is a convex sum of stabilizer states, and the robustness measure satisfies some other useful properties in general resource theories~\cite{Chitambar2019quantum}.

We now show that the optimal $\sigma$-assisted SO discrimination probability is upper bounded by the robustness of magic as follows:
\begin{theorem} \label{thm:robustness_bound} [Robustness bound on  $\sigma$-assisted SO-discrimination probability] For a binary ensemble ${\cal E}^{\rm bin}$, the $\sigma$-assisted SO discrimination probability is bounded by
\begin{equation}
    p_{\rm succ}^{\rm SO} ({\cal E}^{\rm bin};\sigma) \leq \frac{1}{2} + {\cal R}(\sigma)  \left( p_{\rm succ}^{\rm SO} ({\cal E}^{\rm bin}) - \frac{1}{2} \right),
    \label{eq:robustness_success_bound_bin}
\end{equation}
or equivalently, in terms of the $\sigma$-assisted SO data-revealing ratio,
\begin{equation} \label{eq:robustness_hiding_bound}
	R_{\rm reveal}({\cal E}^{\rm bin}; \sigma) := \frac{p_{\rm succ}^{\rm SO}({\cal E}^{\rm bin}; \sigma) - \frac{1}{2} }{p_{\rm succ}^{\rm SO}({\cal E}^{\rm bin}) - \frac{1}{2} } \leq {\cal R}(\sigma).
\end{equation}
For a general ensemble ${\cal E}$, we have a weaker bound
\begin{equation}
    \frac{p_{\rm succ}^{\rm SO} ({\cal E};\sigma)}{p_{\rm succ}^{\rm SO} ({\cal E})} \leq {\cal R}(\sigma) .
    \label{eq:robustness_success_bound}
\end{equation}
\end{theorem}
The basic idea of the proof follows Ref.~\cite{Takagi19general} for general resource theories, while we are restricting the ensemble to be free states (see Appendix~\ref{app:robustness} for a complete proof).

Theorem~\ref{thm:robustness_bound} shows that the nonstabilizerness test proposed in section~\ref{subsec:verification} can also quantitatively estimate the robustness of the injected nonstabilizer state $\sigma$ as given in Eq.~\eqref{eq:robustness_hiding_bound}. We also note that the robustness bound also provides a bound on the necessary amount of resources for perfect discrimination of stabilizer states.
\begin{corollary} [Robustness bound for perfect discrimination] For a given binary ensemble ${\cal E}^{\rm bin}$, the required amount of robustness of the ancillary state $\sigma$ for perfect SO discrimination is lower bounded by the data-hiding ratio as
\begin{equation}
	{\cal R}(\sigma) \geq R_{\rm hiding}({\cal E}^{\rm bin}).
\end{equation}
\end{corollary}
\noindent This immediately follows from Theorem~\ref{thm:robustness_bound} by taking $p_{\rm succ}^{\rm SO} ({\cal E}^{\rm bin};\sigma) =1$ for perfect discrimination.

As an illustrative example, we use the robustness bound to estimate the number of $T$ states required for perfect discrimination of the three-qubit binary ensemble ${\cal E}_3^{\rm bin}$. The robustness of a single $T$ state is ${\cal R}(\psi_T)=\sqrt{2}$. The exact two- and three-copy values are ${\cal R}\left(\psi_T^{\otimes 2}\right) = \frac{1+3\sqrt{2}}{3} \simeq 1.7476 \leq {\cal R}(\psi_T)^2$ and $ {\cal R}\left(\psi_T^{\otimes 3}\right) = \frac{1+4\sqrt{2}}{3} \simeq 2.2190 \leq {\cal R}(\psi_T)^3$,  respectively, where the inequalities follow from the submultiplicativity of robustness~\cite{Howard17application}. Consequently, since $ {\cal R}\left(\psi_T^{\otimes 2}\right) < R_{\rm hiding}\left({\cal E}_3^{\rm bin}\right) = 2 < {\cal R}\left(\psi_T^{\otimes 3}\right)$, at least three $T$ states are required for perfect discrimination of ${\cal E}_3^{\rm bin}$. However, exhaustive optimization over all possible SO protocols gives $p_{\rm succ}^{\rm SO} \left( {\cal E}_3^{\rm bin}; \psi_T^{\otimes 3} \right) \simeq 0.92678 < 1$ (see also figure~\ref{fig:magic_bound}), demonstrating that the robustness bound is not tight for this ensemble. Four $T$ states, on the other hand, suffice for perfect discrimination because they can be used to implement a single Toffoli gate~\cite{Jones2023low}, which enables ${\cal E}_3^{\rm bin}$ to be discriminated perfectly. Thus, the exact minimum number of $T$ states required for perfect discrimination is four.

Nevertheless, the robustness bound provides a readily computable lower bound on the required amount of nonstabilizerness, particularly for larger systems, for which exhaustive optimization over all possible SO protocols becomes computationally demanding. In particular, the number of $T$ states required for the $M$-XOR guessing task can be estimated as follows.

\begin{corollary}
Assuming that Conjecture~\ref{conj:M-XOR_3q} holds, at least $\lceil 2(M + 1 -\log_2(1/\epsilon)) \rceil$ copies of the $T$ state are required to achieve a discrimination success probability of $\frac{1}{2}+\epsilon$.
\end{corollary}

This result follows directly by setting the target success probability to $\frac{1}{2}+\epsilon$ in the robustness bound and applying the submultiplicativity of robustness, ${\cal R}(\psi_T^{\otimes t}) \leq {\cal R}(\psi_T)^t$. Moreover, even if only a weaker exponential decay of the advantage over random guessing can be established for the general $M$-XOR guessing task, perfect discrimination still requires a number of $T$ states that grows at least linearly with $M$.

As an alternative way to increase the $T$-count required for discrimination, one may consider the generalized $n$-qubit ensembles defined in Eq.~\eqref{eq:n-qubit_states}, whose perfect discrimination may require the implementation of Boolean control functions. Compiling such nonlinear functions into the Clifford+$T$ gate set generally introduces multicontrolled gates and entails nontrivial non-Clifford gate counts and circuit depths~\cite{amy2013meet,amy2014polynomial,gidney2018halving}.

\section{Implications in quantum information theory}
\subsection{Stabilizer no-cloning theorem}
\label{subsec:stabilizer_no_cloning}

Quantum state discrimination is closely connected to the no-cloning theorem~\cite{barnett2009quantum, bae2015quantum}. The conventional no-cloning theorem prohibits the perfect cloning of an unknown state drawn from a set containing nonorthogonal states~\cite{Wooters82single}. If such cloning were possible, the cloning process could be applied repeatedly to generate arbitrarily many copies, allowing the nonorthogonal states to be distinguished with arbitrarily high accuracy. In a similar manner, our construction of stabilizer states that are indistinguishable under SO directly implies a strong restriction on the cloning capabilities of SO, even when the unknown state is drawn from a set of mutually orthogonal stabilizer states.

Important distinctions between our construction and the conventional no-cloning theorem are that a given ensemble ${\cal E}$ consists of a finite number of quantum states and is mutually orthogonal; a given finite copy of quantum states may allow for perfect discrimination under restricted operations. By contrast, the discrimination probability of nonorthogonal states only approaches $1$ as the number of copies increases. In particular, we observe that two copies of the states in ${\cal E}_{\rm 3}$ can be perfectly discriminated by SO, i.e., $p_{\rm succ}^{\rm SO} ( \{ \psi \otimes \psi \}_{\psi \in {\cal E}_3} ) =1$. This can be easily done by first performing the optimal Pauli-decision tree measurement in figure~\ref{fig:states}(b) for the first copy. This leads to uncertain results only for the sign information of $\psi_A^\pm$ or $\psi_{\rm GHZ}^\pm$ and all other states could be perfectly determined. Hence, if the first copy's measurement is uncertain about the sign, then we can determine this sign using the second copy of the states. 

This perfect finite-copy discrimination leads to a general relation between free-state discrimination and cloning for any resource theory characterized by a set of free states $\mathbb{F}$ and a set of free operations $\mathbb{O}$, not only limited to SOs. Let ${\cal E} = \{ \psi_k\}_{k=1}^K \subseteq \mathbb{F}$ be an ensemble of mutually orthogonal pure free states, which can be prepared by free operation $\mathbb{O}$. We define the discrimination probability under free operation $\mathbb{O}$ as
\begin{equation}\label{eq:general_discrim}
	p_{\rm succ}^{\mathbb{O}}({\cal E}) := \frac{1}{K} \max_{{\cal M}_\mathbb{O}} \sum_{k=1}^K {\rm Tr}[ \psi_k M_k ],
\end{equation}
where ${\cal M}_\mathbb{O} = \{ M_k \}_{k=1}^K$ is a measurement that can be realized by $\mathbb{O}$. We also define the optimal fidelity of $m$-copy cloning, $\psi_k \xrightarrow{\Lambda} \psi_k^{\otimes m}$, using the allowed operation $\Lambda \in \mathbb{O}$ as
\begin{equation}\label{eq:cloning_fidelity}
	F_{m}^{\mathbb O}({\cal E}) := \frac{1}{K} \max_{\Lambda\in{\mathbb{O}}} \sum_{k=1}^K {\rm Tr} \left[ \psi_k^{\otimes m} \Lambda(\psi_k) \right].
\end{equation}
We then have the following relation between two quantities:
\begin{theorem}[Cloning--discrimination equivalence under free operations]
\label{prop:cloning_discrimination}
Suppose that $r$ copies of the states in ${\cal E} = \{ \psi_k \}_{k=1}^K$ can be perfectly discriminated by free operation $\mathbb{O}$, i.e., $p_{\rm succ}^{\mathbb{O}} \left( \{ \psi_k^{\otimes r} \}_{k=1}^K \right) = 1 $. Then, for every $m\geq r$, 
\begin{equation}
    F_{m}^{\mathbb{O}}({\cal E}) = p_{\rm succ}^{\mathbb{O}}({\cal E}).
    \label{eq:cloning_discrimination_identity}
\end{equation}
\end{theorem}

\begin{proof}
We first note that if $r$-copy of states in ${\cal E} = \{ \psi_k \}_{k=1}^K$ are perfectly discriminated, then the states $\psi_k$ should be mutually orthogonal to each other.

We then show that $F_{m}^{\mathbb{O}}({\cal E}) \leq p_{\rm succ}^{\mathbb{O}}({\cal E})$. Let $\{M_k^{(m)}\}_{k=1}^{K}$ be a measurement that can be realized by $\mathbb{O}$ that perfectly distinguishes the ideal $m$-copy states $\left\{ \psi_k^{\otimes m} \right\}_{k=1}^K$. We may focus on the case with $m=r$, as for $m>r$, the additional copies may simply be ignored for perfect discrimination. The perfect $m$-copy discrimination condition becomes ${\rm Tr} \left[ M_k^{(m)}\psi_k^{\otimes m} \right] = \bra{\psi_k}^{\otimes m} M_k^{(m)}\ket{\psi_k}^{\otimes m}= 1$ for all $k = 1, 2, \cdots, K$, and as $M_k^{(m)}$ are measurement operators, it follows that $ M_k^{(m)} \geq \ket{\psi_k}^{\otimes m}\bra{\psi_k}^{\otimes m} M_k^{(m)}\ket{\psi_k}^{\otimes m}\bra{\psi_k}^{\otimes m} = \left(\ket{\psi_k}\bra{\psi_k}\right)^{\otimes m} =  \psi_k^{\otimes m}$. Consequently, for the optimal free operation $\Lambda^* \in \mathbb{O}$ for $m$-copy cloning,
\begin{equation}\label{eq:clon_disc_lower}
\begin{aligned}
	F_{m}^{\mathbb O}({\cal E}) &=  \frac{1}{K} \sum_{k=1}^K {\rm Tr} \left[ \psi_k^{\otimes m}\Lambda^*(\psi_k) \right]\\
	& \leq \frac{1}{K} \sum_{k=1}^K {\rm Tr} \left[ M_k^{(m)}\Lambda^*(\psi_k) \right] \\
	&\leq p_{\rm succ}^{\mathbb{O}}({\cal E}),
\end{aligned}
\end{equation}
where the last inequality follows from that both $\Lambda^*$ and $\{M_k^{(m)}\}_{k=1}^{K}$ are free operations.

Conversely, $F_{m}^{\mathbb{O}}({\cal E}) \geq p_{\rm succ}^{\mathbb{O}}({\cal E})$ as one can construct a cloning channel by performing an optimal one-copy discrimination protocol $\{ M_k^* \}$ using free operation, followed by preparing $\psi_k^{\otimes m}$ based on the guessed label $k$, i.e.,
\begin{equation}
	\Lambda(\psi_k) = p_k \psi_k^{\otimes m} + (1-p_k)\rho_{k_\perp},
\end{equation}
where $p_k = {\rm Tr}\left[ \psi_k M_k^* \right]$ is the correct guessing probability of label $k$ and $\rho_{k_\perp}$ refers to the state with failed guess, which is orthogonal to $\psi_k$ as every state in ${\cal E}$ is orthogonal to each other. Consequently, we obtain
\begin{equation}\label{eq:clon_disc_upper}
\begin{aligned}
	F_{m}^{\mathbb O}({\cal E}) &\geq \frac{1}{K} \sum_{k=1}^K {\rm Tr} \left[ \psi_k^{\otimes m} \Lambda(\psi_k) \right]\\
	 &= \frac{1}{K} \sum_{k=1}^K p_k \\
	 &= \frac{1}{K} \sum_{k=1}^K {\rm Tr}\left[ \psi_k M_k^* \right] \\
	 &= p_{\rm succ}^{\mathbb{O}}({\cal E}).
\end{aligned}
\end{equation}

Hence, combining Eqs.~\eqref{eq:clon_disc_lower} and \eqref{eq:clon_disc_upper}, we have $F_{m}^{\mathbb O}({\cal E}) = p_{\rm succ}^{\mathbb{O}}({\cal E})$, proving the equality.
\end{proof}

In our focus of SO-indistinguishable stabilizer states, the perfect SO-discrimination of two-copy ${\cal E}_3$ is possible, and the same protocol applies for the two-copy six-state ensemble ${\cal E}_{\rm 3,SEP}$. We then have a following no-cloning theorem for stabilizer states:
\begin{corollary}[Stabilizer no-cloning theorem]
\label{cor:stabilizer_no_cloning}
No stabilizer operation can perfectly clone every state in the mutually orthogonal basis ${\cal E}_{\rm 3}$ with the average fidelity
\begin{equation}
    F_{m}^{\rm SO} \left( {\cal E}_{\rm 3} \right) = p_{\rm succ}^{\rm SO} \left( {\cal E}_{\rm 3} \right) = \frac{3}{4},~\forall m \geq 2.
    \label{eq:E3q_cloning_fidelity}
\end{equation}
Similarly, for the six-state separable states ensemble,
\begin{equation}
    F_{m}^{\rm SO} \left( {\cal E}_{\rm 3,SEP} \right) = \frac{5}{6},~ \forall m\geq2.
    \label{eq:E3q_SEP_cloning_fidelity}
\end{equation}
\end{corollary}

The perfect finite-copy discrimination also allows us to quantify the magic required for stabilizer cloning. Suppose that an ancillary state
$\sigma$ is supplied, while all remaining operations are restricted to SO.
Define the corresponding resource-assisted global cloning fidelity as
\begin{equation}
    F_{m}^{\rm SO} ({\cal E};\sigma) := \frac{1}{K} \max_{\Lambda\in{\rm SO}} \sum_{k=1}^K \operatorname{Tr} \left[ \psi_k^{\otimes m} \Lambda(\psi_k\otimes\sigma) \right].
    \label{eq:assisted_cloning_fidelity}
\end{equation}
Here, $\Lambda$ acts jointly on the input state, the ancillary resource
$\sigma$, and arbitrary stabilizer ancillas. Whenever $r$-copies can be perfectly discriminated by SO, the cloning--discrimination relation extends directly to the resource-assisted setting as 
\begin{equation}
    F_{m}^{\rm SO} ({\cal E};\sigma) = p_{\rm succ}^{\rm SO} ({\cal E};\sigma),~\forall m \geq r.
    \label{eq:assisted_cloning_discrimination}
\end{equation}
Indeed, composing a resource-assisted cloner with the perfect $m$-copy SO discriminator gives a resource-assisted one-copy discrimination protocol. Conversely, one may first perform an optimal resource-assisted discrimination and then prepare $m$ copies of the guessed stabilizer state. Thus, for the ensembles considered here, the magic cost of perfect cloning is identical to the magic cost of perfect discrimination. 

\begin{table*}[!t]
\begin{center}
\begin{tabular}{ c | c | c }
\hline
Resource theory & Entanglement  & Magic \\
 \hline \hline
Free states & $\begin{matrix} \text{Separable} \\ \text{states} \end{matrix}$ & $\begin{matrix} \text{Stabilizer} \\ \text{states} \end{matrix}$ \\  
\hline
$\begin{matrix} \text{Free operations} \\ \text{(Operational)} \end{matrix}$ & LOCC & $\begin{matrix} \text{Stabilizer} \\ \text{operations} \end{matrix}$    \\
\hline
$\begin{matrix} \text{Free operations} \\ \text{(Axiomatic)} \end{matrix}$ & $\begin{matrix} \text{Separable} \\ \text{operations (SEP)} \end{matrix}$ & $\begin{matrix} \text{Completely stabilizer-preserving} \\ \text{operations (CSPO)} \end{matrix}$ \\
\hline 
$\begin{matrix} \text{Related free-state} \\ \text{discrimination task}\end{matrix}$ & $\begin{matrix} \text{Nonlocality without} \\ \text{entanglement~\cite{Bennet99quantum}} \end{matrix}$ & $\begin{matrix} \text{SO-indistinguishable} \\ \text{ stabilizer states (this work)} \end{matrix}$ \\\hline
Implication & LOCC $\subsetneq$ SEP\cite{Bennet99quantum} & SO $\subsetneq$ CSPO (this work)
\\\hline
\end{tabular}
\end{center}
\caption{Comparison between the resource theories of entanglement and magic. Related discrimination tasks provide a strict separation between operational and axiomatic free operations in both resource theories.}
\label{tab:ent_mag}
\end{table*}

\subsection{Relation between ``nonlocality without entanglement''}
Our main observation reveals a fundamental asymmetry between the cost of preparing stabilizer states and that of discriminating them, leading to several interesting implications for quantum resource theory and quantum foundations. This parallels the analogous phenomenon in entanglement theory known as ``nonlocality without entanglement''~\cite{Bennet99quantum}.\footnote{In this sense, one might refer to the phenomenon studied in this work as ``nonstabilizerness without magic.'' However, because the term \emph{magic} is used to describe both states and operations, this phrase may cause confusion. We therefore adopt the more direct term ``SO-indistinguishable stabilizer states,'' which conveys the intended meaning more clearly.}

A direct consequence is a strict separation between the operational class of stabilizer operations and the axiomatic class of completely stabilizer-preserving operations (CSPOs)~\cite{seddon2019quantifying,Seddon21quantifying} in the quantum resource theory of magic:
\begin{equation}
    {\rm SO} \subsetneq {\rm CSPO}.
\end{equation}

\begin{corollary}[Strict SO--CSPO separation]
The quantum channel
\begin{equation}
	\Phi (\rho) = \sum_{\psi \in {\cal E}_3} \Pi_{\psi} \rho \Pi_{\psi} \otimes \ket{\boldsymbol{x}(\psi)}\bra{\boldsymbol{x}(\psi)},
\end{equation}
where $\Pi_{\psi} =\ket{\psi} \bra{\psi}$ denotes the projector onto a state $\psi$ in the three-qubit basis ${\cal E}_3$ and $\boldsymbol{x}(\psi) \in \{0, 1\}^3$ is a corresponding labels for each $\psi$, is a CSPO but not an SO.
\end{corollary}

It is straightforward to show that $\Phi$ is a CSPO, since projection onto stabilizer states preserves stabilizer states even when the channel acts on one part of a larger system. On the other hand, if $\Phi$ were an SO, measuring the classical outcome associated with each projector $\Pi_\psi$ would perfectly discriminate the states in ${\cal E}_3$ using stabilizer operations, contradicting their SO-indistinguishability. Although stabilizer operations were already shown to form a strict subset of CSPOs in Ref.~\cite{heimendahl2022axiomatic}, our result provides a concrete example with a clear operational interpretation in quantum state discrimination. This is analogous to how nonlocality without entanglement separates LOCC from separable operations (see table~\ref{tab:ent_mag} for a comparison between the resource theories of entanglement and magic).

An interesting observation concerning the relationship between SO-indistinguishable stabilizer states and LOCC-indistinguishable separable states is the existence of ensembles of separable stabilizer states. Remarkably, all the separable-state examples considered in this work, namely ${\cal E}_{\rm 3,SEP}$, ${\cal E}_{\rm 2}$~\cite{Bennet99quantum}, and ${\cal E}_{\rm 3,min}$~\cite{Bennett1999unextendible, Duan2010locally}, have also been shown to be LOCC-indistinguishable. This observation can be generalized to ensembles of pure product stabilizer states as follows.
\begin{theorem}\label{prop:SO_LOCC}
Let ${\cal E}=\{ \psi_k \}_{k=1}^K$ be an ensemble of $n$-qubit pure product stabilizer states, $\ket{\psi_k} = \bigotimes_{j=1}^n \ket{\phi_{k,j}}$, where each $\ket{\phi_{k,j}}$ is a single-qubit stabilizer state. If ${\cal E}$ is SO-indistinguishable, then it cannot be perfectly distinguished by any finite-round LOCC protocol.
\end{theorem}

The proof is provided in Appendix~\ref{app:stab_LOCC_prop}. For general ensembles of separable stabilizer states, we leave the following statement as an open conjecture.
\begin{conjecture}\label{conj:SO_LOCC}
Let ${\cal E}=\{\rho_k\}_{k=1}^K$ be an ensemble of separable stabilizer states. Then ${\cal E}$ is SO-indistinguishable if and only if it is LOCC-indistinguishable.
\end{conjecture}
If Conjecture~\ref{conj:SO_LOCC}  turns out to be true, it may reveal a deeper connection between the resource theories of entanglement and magic, despite their distinct structures. A weaker version of the conjecture could be formulated for complete stabilizer bases as in ${\cal E}_3$ or unextendible product bases~\cite{Bennett1999unextendible,bravyi2004unextendible,Duan2010locally} as in ${\cal E}_{\rm 3, min}$, and it may remain true even if Conjecture~\ref{conj:SO_LOCC} turns out to be false.

\section{Remarks}
We have demonstrated the existence of a set of stabilizer states that cannot be perfectly discriminated using only SOs. Although stabilizer states and stabilizer circuits admit efficient classical descriptions and simulation, our result reveals a sharp contrast between the tractability of the stabilizer framework and its operational limitations for state discrimination. It therefore provides new insight into the operational structure of the resource theory of magic. Our result also enables a new form of data hiding against SOs. Building on this primitive, the nonstabilizerness of a quantum process can be verified by asking a prover to reveal information that would remain inaccessible under the restriction to SOs. The resulting protocol not only provides a verification of nonstabilizerness using a small number of trials without relying on computational hardness assumptions, but also yields a quantitative bound on the robustness of magic within the model of quantum computation based on magic-state injected Clifford circuits. In this manner, our protocol could provide an efficient method for benchmarking non-Clifford gates~\cite{Erhad19characterizing, Hines23demonstrating, lee2025efficient}. In fundamental aspects, our result rules out cloning the stabilizer states using only SOs, where the underlying argument may extend to other quantum resource theories. Furthermore, the separation between axiomatic and operational notions of free operations in the resource theory of magic parallels the distinction between separable operations and LOCC in entanglement theory, which may imply a deeper connection between two apparently distinct quantum resource theories. Together, these results establish restricted state discrimination as a useful operational probe of the boundaries of the stabilizer framework.

Our findings suggest several directions for future research. Deriving SO distinguishability bounds for multiqubit stabilizer states and characterizing the optimal discrimination strategies would deepen our structural understanding of stabilizer-state discrimination. It would also be worthwhile to investigate whether catalytic access to non-Clifford resource states can improve discrimination performance~\cite{kim2025catalytic}.  Addressing these questions may ultimately clarify how magic functions as an operational resource in early fault-tolerant quantum algorithms~\cite{Katabarwa24early}. The parallels observed here between the resource theories of entanglement and magic raise a broader question in quantum resource theory: under what conditions can free states be discriminated in principle yet remain indistinguishable under an operationally motivated class of free operations? While a general condition for such a separation has recently been investigated~\cite{ao2026witness}, uncovering deeper structural connections among different resource theories remains an important future direction. The connection proposed in Conjecture~\ref{conj:SO_LOCC} provides one concrete example. It would also be interesting to determine whether the SO-indistinguishability of stabilizer states can give rise to unconditional advantages in computational tasks.

\acknowledgements{The author thanks Ryuji Takagi, Bartosz Regula, Seok Hyung Lie, Minki Hhan, and Francesco Anna Mele for constructive discussions and insightful advice throughout this project. This work is supported by the KIAS Individual Grant No. CG085302 at Korea Institute for Advanced Study and  the National Research Foundation of Korea (NRF) grant funded by the Korea government (MSIT) (No. RS-2024-00413957).}	

\section*{Declaration of AI use}
During the preparation of this manuscript, ChatGPT-6 Sol was used to write code implementing dynamic programming for adaptive Pauli-decision trees, to explore approaches to proving Theorem~\ref{prop:SO_LOCC}, and to assist with proofreading and checking the consistency of the proofs. All code and suggested revisions were reviewed and, where appropriate, edited by the author.
\bibliography{bib.bib}

~\newpage
\widetext
\begin{appendix}
\section{Complete proof of Theorem~\ref{thm:main}} \label{app:thm1}
We note that the effect of any Clifford operation acting on an $n$-qubit state followed by a computational basis measurement on a single-qubit with outcome $a \in \{0, 1\}$ can be represented as
\begin{equation}
\Pi_a^P = \frac{\mathbb{1} + (-1)^a P}{2},
\end{equation}
for some hermitian Pauli operator $P \in {\cal P}_n$ with $P^\dagger = P$. This can be shown by noting that a Clifford unitary $U \in {\rm Cl}_n$ followed by computational basis measurement on the $j$-th qubit $\Pi_a^{(j)} = \mathbb{1}\otimes \cdots \otimes \ket{a}_j \bra{a} \otimes \cdots \otimes \mathbb{1}= \frac{\mathbb{1} + (-1)^a Z_j}{2}$ with outcome $a \in \{ 0 , 1 \}$ can be expressed as $\Pi_a^{(j)} U \rho U^\dagger \Pi_a^{(j)}= U U^\dagger \Pi_a U \rho U^\dagger \Pi_a U U^\dagger =  U \Pi_a^P \rho \Pi_a^P U^\dagger$ with $P = U^\dagger Z_j U$, where $Z_j$ is the $Z$ operator acting on the $j$-th qubit. Then, the remaining $U$ after the measurement effect can be absorbed into the next step of the stabilizer operations. Also, multi-qubit measurements can be regarded as a sequence of single-qubit measurements. This leads to the fact that any stabilizer operation can be decomposed into a sequence of such effects, i.e.,
\begin{equation}
\Pi_{\boldsymbol{a}} = \prod_{j=1}^L \Pi_{a_j}^{P_j},
\end{equation}
where $L$ can be arbitrarily large, and each $P_i$ can be chosen adaptively based on the history of previous outcomes $(a_1, a_2, \ldots, a_j)$ with $j < i$. While any stabilizer operations can be, in principle, realized in a finite number of rounds~\cite{heimendahl2022axiomatic}, this does not affect the proof of our statement.

Now let us decompose the discrimination process $\Pi_{\boldsymbol{a}}$ into two parts:  
\begin{enumerate}
\item[i)] the first measurement, $\Pi_{a}^{P} = \Pi_{a_1}^{P_1}$ 
\item[ii)] all remaining measurements, $\tilde\Pi_{\bar{\boldsymbol{a}}} := \prod_{j=2}^L \Pi_{a_j}^{P_j}$.
\end{enumerate}
The post-measurement states after the first measurement then become
\begin{equation}
\ket{\psi_k^{P,a}} \propto \Pi_a^P \ket{\psi_k},
\end{equation}
with probability $p(a | k) = \bra{\psi_k} \Pi_a^P \ket{\psi_k} = \frac{1}{2} \left( 1 + (-1)^a \bra{\psi_k} P \ket{\psi_k} \right)$. From the fact that the expectation value of any Pauli operator $P$ for stabilizer states is $\{ -1 , 0, 1 \}$, the measurement probability can only take three possible values of $p(a | k) \in \left\{ 0, \frac{1}{2}, 1 \right\}$.

We then show that, for any non-trivial first-round measurement $\Pi_a^P$ with $P \neq \mathbb{1}$, there exists $k \neq k'$ satisfying the following two conditions:
\begin{enumerate}
\item $p(a | k) = \frac{1}{2} \Leftrightarrow \bra{\psi_k} P \ket{\psi_k} = 0,$
\item $\left| \langle \psi_{k'}^{P,a} | \psi_{k}^{P,a} \rangle \right|^2  > 0 \Leftrightarrow \left| \bra{\psi_{k'}} P \ket{\psi_{k}} \right|^2 > 0.$
\end{enumerate}
The first condition means that $\ket{\psi_k}$ cannot be unambiguously distinguished by $\Pi_a^P$, and the second condition implies that the post-measurement states become non-orthogonal. In other words, for any measurement $\Pi_a^P$, there exists a pair of quantum states $\ket{\psi_k}$ and $\ket{\psi_{k'}}$ that cannot be fully determined. Moreover, the post-measurement states after the first measurement cannot be perfectly distinguished by any (possibly adaptive) sequential operations, even if nonstabilizer operations are allowed, because no quantum measurement can fully distinguish a set of non-orthogonal quantum states.

Let us first show that the set of stabilizer states
\begin{equation}\label{eq:app_set}
\begin{aligned}
{\cal E}_{\rm 3, SEP}  = \{ \psi_k \}_{k=1}^6 =\left\{ \psi_A^{\pm}, \psi_B^{\pm}, \psi_C^{\pm} \right\} =
\{ \ket{\pm 1 0}, \ket{0 {\pm} 1}, \ket{1 0 \pm} \}
\end{aligned}
\end{equation}
cannot be perfectly discriminated by any stabilizer operations without ancillary qubits. Let us define the stabilizer of a given stabilizer state $\ket{\psi}$ as ${\cal S}_{\psi}$. Sets of stabilizer generators for each state in Eq.~\eqref{eq:app_set} are
\begin{equation}
\begin{aligned}
{\cal S}_{\psi_A^{\pm}} &= \langle \pm X_1, - Z_2, Z_3 \rangle, \\
{\cal S}_{\psi_B^{\pm}} &= \langle Z_1, \pm X_2, -Z_3 \rangle, \\
{\cal S}_{\psi_C^{\pm}} &= \langle -Z_1, Z_2, \pm X_3 \rangle.
\end{aligned}
\end{equation}
Since $\bigcup_{\psi \in {\cal E}_{\rm 3, SEP}}  {\cal S}_{\psi}$ contains all single-qubit Pauli operators $X_j$ and $Z_j$ for $j =1,2,3$, any non-trivial $P \neq \mathbb{1}$ anticommutes with at least one element in $\bigcup_{\psi \in {\cal E}_{\rm 3, SEP}}  {\cal S}_{\psi}$ .

We then consider two cases to cover all possible $P \neq \mathbb{1}$: (i) $P$ anticommutes with at least one of $Z_{1,2,3}$, and (ii) $P \neq \mathbb{1}$ commutes with all $Z_{1,2,3}$, so that it must anticommute with at least one of $X_{1,2,3}$. For each case, we have
\begin{enumerate}
\item[(i)] If $P Z_j = - Z_j P$ for some $j=1,2,3$, let $\ket{\psi_k}$ be the state stabilized by $Z_j$ and $\ket{\psi_{k'}}$ be the state stabilized by $-Z_j$. For example, if $P$ anticommutes with $Z_2$, we take $\ket{\psi_k} = \ket{10+}$ and $\ket{\psi_{k'}} = \ket{+10}$. Then, it is straightforward to check that $P\ket{\psi_k}$ and $\ket{\psi_{k'}}$ are non-orthogonal, as the $j$th qubits of these states are the same and the other qubits are in different $X$ and $Z$ bases, giving a non-zero overlap: $\langle \pm | 0 \rangle \neq 0$ and $\langle \pm | 1 \rangle \neq 0$.
\item[(ii)] If $P Z_{1,2,3} = Z_{1,2,3} P$ and $P X_j = - X_j P$ for some $j=1,2,3$, take two states $\ket{\psi_k}$ and $\ket{\psi_{k'}}$ that are stabilized by $X_j$ and $-X_j$, respectively. For example, if $P$ anticommutes with $X_1$, we take $\ket{\psi_k} = \ket{+10}$ and $\ket{\psi_{k'}} = \ket{-10}$. Since the remaining stabilizers of these two states are identical $Z$ operators, we have a perfect overlap between the two states: $P \ket{\psi_k} = \ket{\psi_{k'}}$.
\end{enumerate}
Hence, for any $P \neq \mathbb{1}$, one can always find a pair of quantum states $(\ket{\psi_k}, \ket{\psi_{k'}})$ such that $\bra{\psi_k} P \ket{\psi_k} = 0$ and $\bra{\psi_{k'}} P \ket{\psi_k} \neq 0$, which implies $|\bra{\psi_{k'}} P \ket{\psi_k}|^2 > 0$.

Now, let us consider the case with additional ancillary qubits. We note that any stabilizer-state ancilla can be expressed as $ U \ket{0}^{\otimes \ell}$ with a Clifford unitary $U\in {\rm Cl}_\ell$, which can be absorbed into the measurement operator $\Pi_a^P$. Hence, without loss of generality, we can take the ancillary qubits to be $\ket{0}^{\otimes \ell}$, and then prove the following lemma:

\begin{lemma} [see also Lemma 5 in Ref.~\cite{heimendahl2022axiomatic}]\label{lemma:stab_recover} 
Suppose that a set of quantum states ${\mathcal X}_Q = \{ \ket\phi \}$ is stabilized by a Pauli operator $Q$, i.e., $Q \ket{\phi} = \ket{\phi}$ for all $\ket{\phi} \in {\mathcal X}_Q$. Then, the effect of any measurement operator $\Pi^{P}_a = \frac{\mathbb{1} + (-1)^a P}{2}$ acting on ${\mathcal X}_Q$ can be undone by a Clifford operation whenever $P Q = - QP$, i.e.,
\begin{equation}
	U_a^{P,Q} \Pi^{P}_a \ket{\phi} \propto \ket{\phi}, \quad \forall \ket{\phi} \in {\mathcal X}_Q, a \in \{ 0, 1\},
\end{equation}
where 
\begin{equation}
	U_a^{P,Q} :=\frac{Q+(-1)^aP}{\sqrt 2} = Q e^{-i (-1)^a \frac{\pi}{4} R },
\end{equation}
with a Hermitian Pauli operator $R := i Q P $.
\end{lemma}
\begin{proof}
It is straightforward to check that $U_a^{P,Q}$ is a Clifford unitary operator, as it is a combination of a $\pi/2$ rotation over the Pauli operator $R$ followed by the Pauli unitary $Q$. We now directly compute 
\begin{equation}
\begin{aligned}
	U_a^{P,Q} \Pi^{P}_a &=  \left( \frac{Q+(-1)^aP}{\sqrt 2} \right)\left( \frac{\mathbb{1} + (-1)^a P}{2} \right)\\
	&= \frac{\mathbb{1} + Q  + (-1)^a P + (-1)^a Q P }{2\sqrt{2}} \\
	&= \frac{\mathbb{1} + Q  + (-1)^a P - (-1)^a PQ }{2\sqrt{2}} \\
	&= \frac{1}{\sqrt{2}} \left[ \left( \frac{\mathbb{1} + Q}{2}\right) + (-1)^a P \left( \frac{\mathbb{1} - Q}{2}\right)  \right].
\end{aligned} 
\end{equation}
We then note that this operator has a trivial effect when applied to any quantum state $\ket\psi$ stabilized by $Q$, i.e., $Q \ket\psi = \ket{\psi}$, as
\begin{equation}
	U_a^{P,Q} \Pi^{P}_a \ket{\phi} = \frac{1}{\sqrt{2} }\left[ \left( \frac{\mathbb{1} + Q}{2} \right) + (-1)^a P \left( \frac{\mathbb{1} - Q}{2} \right) \right] \ket\phi = \frac{1}{\sqrt{2}} \ket\phi \propto \ket{\phi},
\end{equation}
which completes the proof. We note that the square of the normalization factor $1/\sqrt{2}$ corresponds to the outcome probability of $\Pi_a^P$, which is balanced for the states $Q \ket\phi = \ket\phi$.
\end{proof}

Now let us complete the proof by showing that appending $\ell$ ancilla states to a given $n$-qubit state $\ket\psi$ as
\begin{equation}
	\ket\psi \otimes \ket{0}^{\otimes \ell},
\end{equation}
does not help. We note that the state is stabilized by $Z_{n+j}$ for $j = 1, \cdots, \ell$. Following Lemma~\ref{lemma:stab_recover}, any stabilizer operation having a non-trivial effect on $\ket\psi \otimes \ket{0}^{\otimes \ell}$ should be in a form of $\Pi_a^P$ with $P$ that commute with every $Z_{n+j}$, i.e. $P = P^{(n)} \otimes \left[ \bigotimes_{j=1}^\ell P_j\right]$ with $P_j \in \{ \pm \mathbb{1}, \pm Z \}$. However, this leads to trivial effects on the ancilla states as 
\begin{equation}
	\Pi_a^P \left( \ket{\psi} \otimes \ket{0}^{\otimes \ell} \right) = \left( \Pi_a^{P^{(n)}} \ket{\psi} \right) \otimes \ket{0}^{\otimes \ell},
\end{equation}
up to the sign factor, which does not change the physics. As the Pauli operator $P^{(n)}$ acts only on the $n$-qubit system $\ket\psi$, the situation reduces to the ancilla-free case. Therefore, adding ancilla qubits does not increase the successful discrimination probability $p^{\rm SO}_{\rm succ}$.

Since the ancilla-free case was already shown above, this completes the proof of Theorem~\ref{thm:main}. As adding more orthogonal states makes the discrimination no easier, it immediately follows that ${\cal E}_{\rm 3} = \left\{ \psi_A^{\pm}, \psi_B^{\pm}, \psi_C^{\pm}, \psi_{\rm GHZ}^{\pm} \right\}$ cannot be perfectly distinguished by SO.

\section{Removing a single state from ${\cal E}_{\rm 3,SEP}$} \label{app:six_ness}
Without loss of generality, suppose that the state $\ket{\psi_C^{-}}=\ket{10-}$ is removed from ${\cal E}_{\rm 3,SEP}$. We first measure the third qubit in the computational basis, i.e., ${\Pi_{a_1}^{Z_3}}$. For $a_1=0$, the possible input states are ${\ket{+10},\ket{-10},\ket{10+}}$, whereas for $a_1=1$, they are ${\ket{0{+}1},\ket{0{-}1},\ket{10+}}$. We then measure $Z_2$ if $a_1=0$ and $Z_1$ if $a_1=1$, thereby distinguishing $\ket{10+}$ from the other two candidates. Finally, the remaining pair is distinguished by measuring $X_1$ in the first branch or $X_2$ in the second branch. Thus, the five remaining states can be perfectly discriminated by SO.

\section{Complete proof of Theorem~\ref{thm:2qubit}} \label{app:thm2}
In order to complete the proof of Theorem~\ref{thm:2qubit}, we first show the following lemma:
\begin{lemma} \label{lemma:2q_perfect}
Suppose that two-qubit mutually orthogonal stabilizer states form a complete basis $\{\ket{\psi_k}\}_{k =1}^4$. The stabilizers of these states can always be expressed as
\begin{equation}\label{eq:thm2_app}
\begin{aligned}
{\cal S}_{\psi_1} &= \langle g_1, g_2 \rangle, \\
{\cal S}_{\psi_2} &= \langle g_1, -g_2 \rangle, \\
{\cal S}_{\psi_3} &= \langle -g_1, g_3 \rangle, \\
{\cal S}_{\psi_4} &= \langle -g_1, -g_3 \rangle.
\end{aligned}
\end{equation}
\end{lemma}

\begin{proof}
Without loss of generality, take ${\cal S}_{\psi_1} = \langle g_1, g_2 \rangle = \{ \mathbb{1}, g_1, g_2, g_1 g_2\}$ with $g_1 \neq g_2$. To construct mutually orthogonal states, the other states must have at least one stabilizer $-g$ with $g \in {\cal S}_{\psi_1}$. Again, without loss of generality, take ${\cal S}_{\psi_3} = \langle -g_1, g_3 \rangle = \{ \mathbb{1}, -g_1, g_3, -g_1 g_3 \}$. Then, for $\psi_2$ to be orthogonal to $\psi_1$, we have the following cases:

\begin{enumerate}
\item ${\cal S}_{\psi_2}$ contains $-g_1$: In this case, ${\cal S}_{\psi_2}$ must also contain $-g_3$ to ensure that $\psi_2$ is orthogonal to $\psi_3$, giving ${\cal S}_{\psi_2} = \langle -g_1, -g_3 \rangle = \{\mathbb{1}, -g_1, -g_3, g_1 g_3 \}$. The remaining state $\psi_4$ must be orthogonal to both $\psi_2$ and $\psi_3$, which requires $g_1 \in {\cal S}_{\psi_4}$. To also be orthogonal to $\psi_1$, ${\cal S}_{\psi_4}$ must contain $-g_2$, yielding ${\cal S}_{\psi_4} = \langle g_1, -g_2 \rangle$. By relabeling $\psi_2 \leftrightarrow \psi_4$, we obtain the desired form.

\item ${\cal S}_{\psi_2}$ contains $-g_2$: Then ${\cal S}_{\psi_2}$ must contain $g_1$, $-g_3$, or $g_1 g_3$ to ensure orthogonality with $\psi_3$.
\begin{enumerate}
\item If ${\cal S}_{\psi_2} = \langle g_1, -g_2 \rangle = \{ \mathbb{1}, g_1, -g_2, -g_1 g_2\}$, then to make $\psi_4$ orthogonal to both $\psi_1$ and $\psi_2$, we require $-g_1 \in {\cal S}_{\psi_4}$. Finally, to make $\psi_4$ orthogonal to $\psi_3$, we have ${\cal S}_{\psi_4} = \langle -g_1, -g_3 \rangle$, which is the desired form.
\item If ${\cal S}_{\psi_2} = \langle -g_2, -g_3 \rangle = \{ \mathbb{1}, -g_2, -g_3, g_2g_3\}$, this implies that $g_1, g_2, g_3$ mutually commute. For a two-qubit system, this is only possible when $g_3 = g_2$ or $g_3 = \pm g_1 g_2$, since $g_1 \neq g_2$ as being a stabilizer of $\psi_1$, and $g_3 \neq -g_2$ as ${\cal S}_{\psi_2} = \langle -g_2, -g_3 \rangle$ cannot contain $g_3$.
\begin{enumerate}
\item For $g_3 = g_2$, the stabilizer sets can be rewritten as ${\cal S}_{\psi_1} = \langle g_1, g_2 \rangle$, ${\cal S}_{\psi_2} = \langle -g_2, g_4 \rangle$, and ${\cal S}_{\psi_3} = \langle -g_1, g_2 \rangle$ with another generator $g_4$. Then ${\cal S}_{\psi_4}$ must contain $-g_2$ to be orthogonal to $\psi_1$ and $\psi_3$, and to be orthogonal to $\psi_2$, we have ${\cal S}_{\psi_4} = \langle -g_2, -g_4 \rangle$. After relabeling $g_2 \leftrightarrow g_1$ and $g_4 \rightarrow g_3$, and swapping $\psi_2 \leftrightarrow \psi_3$, we obtain the desired form.
\item For $g_3 = g_1 g_2$, we have ${\cal S}_{\psi_1} = \langle g_1, g_2 \rangle$, ${\cal S}_{\psi_2} = \langle -g_2, - g_1 g_2 \rangle = \langle g_1, -g_2 \rangle$ and ${\cal S}_{\psi_3} = \langle -g_1, - g_2\rangle$. Then the only possible choice for the last state is ${\cal S}_{\psi_4} = \langle -g_1, g_2 \rangle$, which is the form in Eq.~\eqref{eq:thm2_app} by taking $g_3 = -g_2$.
\item For $g_3 = - g_1 g_2$, we have ${\cal S}_{\psi_1} = \langle g_1, g_2 \rangle$, ${\cal S}_{\psi_2} = \langle -g_2, g_1 g_2 \rangle = \langle -g_1, -g_2 \rangle$ and ${\cal S}_{\psi_3} = \langle -g_1, g_2\rangle$. Then the only possible choice for the last state is ${\cal S}_{\psi_4} = \langle g_1, -g_2 \rangle$, which is the form as the previous case after relabeling.
\end{enumerate} 
\item If ${\cal S}_{\psi_2} = \langle -g_2, g_1 g_3\rangle = \{ \mathbb{1}, -g_2, g_1 g_3, -g_1g_2g_3 \}$, this also implies that $g_1, g_2, g_3$ mutually commute. For a two-qubit system, this is only possible when $g_3 = -g_1 g_2$ or $g_3 = \pm g_2$, since $g_1 \neq g_2$ as being a stabilizer of $\psi_1$, and $g_3 \neq g_1 g_2$ as ${\cal S}_{\psi_2} = \langle -g_2, g_1 g_3 \rangle$ cannot contain $g_2$. Both the case reduces to the case (b), thus belongs to the desired form.

\end{enumerate}
\item ${\cal S}_{\psi_2}$ contains $-g_1 g_2$: By relabeling $\tilde g_1 = g_1$, $\tilde g_2 = g_1 g_2$, and $\tilde g_3 = g_3$ leads to ${\cal S}_{\psi_1} = \langle \tilde g_1, \tilde g_2 \rangle$, ${\cal S}_{\psi_3} = \langle -\tilde g_1, \tilde g_3 \rangle$, and $- \tilde g_2 \in {\cal S}_{\psi_2}$, which is equivalent to the case 2.
\end{enumerate}
As the above cases cover all possibilities, this completes the proof.
\end{proof}
We can easily show that the set in Lemma~\ref{lemma:2q_perfect} can be perfectly discriminated by SO. One can first perform the Pauli measurement with the generator $g_1$, i.e., $\Pi_{a_1}^{g_1}$ to decide whether the target state is in $\{ \psi_1, \psi_2 \}$ (for $a_1 = 0$) or $\{ \psi_3, \psi_4\}$ (for $a_1 = 1$). Based on the first measurement outcome, one can take the second Pauli measurement as $\Pi_{a_2}^{g_2}$ (for $a_1 = 0$) or $\Pi_{a_2}^{g_3}$ (for $a_1 = 1$) which perfectly discriminates between two possible states for both cases.

\begin{lemma} \label{lemma:K} Any $K \leq 3$ mutually orthogonal pure stabilizer states can be perfectly discriminated by SO.
\end{lemma}
\begin{proof}
$K=1$ is trivial, and for any two orthogonal pure stabilizer states, i.e., $K=2$, there always exists a Pauli operator $P$ such that $P \ket{\psi_1} = \ket{\psi_1}$ and $P \ket{\psi_2} = -\ket{\psi_2}$. This directly leads to the perfect discrimination from a Pauli measurement $\Pi_{a}^P$. We now provide a proof for $K=3$. By mutual orthogonality of three stabilizer states $\ket{\psi_1}, \ket{\psi_2}, \ket{\psi_3}$, there exist two mutually commuting Pauli operators $P$ and $Q$, such that
\begin{equation}
\begin{aligned}
&P \ket{\psi_1} = \ket{\psi_1}, \quad P \ket{\psi_2} = -\ket{\psi_2} \\
&Q \ket{\psi_1} = \ket{\psi_1}, \quad Q \ket{\psi_3} = -\ket{\psi_3}.
\end{aligned}
\end{equation}
Then, one can perform sequential measurements $ \Pi_{a_2}^Q \Pi_{a_1}^P$, which commutes with each other, and the outcomes $(a_1, a_2)$ can determine the state's label as
\begin{equation}
\begin{aligned}
 (a_1, a_2) &= (0, 0) 	\mapsto \psi_1\\
 (a_1, a_2) &= (0, 1) 	\mapsto \psi_3\\
 (a_1, a_2) &= (1, 0) 	\mapsto \psi_2\\
 (a_1, a_2) &= (1, 1) 	\mapsto \{ \psi_2, \psi_3 \}.
\end{aligned}
\end{equation}
Even for the last case, the possibility of $\ket{\psi_1}$ is excluded, and the projected states remain orthogonal. Hence, one can perform additional Pauli measurements to perfectly discriminate between $\ket{\psi_2}$ and $\ket{\psi_3}$. This leads to the perfect discrimination of the three states.
\end{proof}

Finally, combining Lemmas~\ref{lemma:2q_perfect} and \ref{lemma:K} completes Theorem~\ref{thm:2qubit}.

\section{Entropic analysis}\label{app:entropy_bound}
We further analyze quantitative bounds on the discrimination task of a given ensemble ${\cal E} = \{ p_k, \rho_k \}_{k=1}^K$ in terms of the mutual information $I({\cal K} : {\cal A}) = H({\cal K}) + H({\cal A}) - H({\cal KA})$ between the distributed label $k$ following distribution ${\cal K}$ and the outcomes $\boldsymbol{a}$ from SO-measurement $\Pi_{\boldsymbol{a}} = \prod_{j=1}^L \Pi^{P_j}_{a_j} $, following distribution ${\cal A}$. Here, $H({\cal X}) = -\sum_{x} p(x) \log_2 p(x)$ is the Shannon entropy of the variable $x$ following distribution ${\cal X}$. The joint distribution is given by $p(k, \boldsymbol{a}) = p_k \, p(\boldsymbol{a} | k)$, and $p(\boldsymbol{a} | k) = {\rm Tr} [ \Pi_{\boldsymbol{a}} \rho_k] $ represents the outcome probability of $\boldsymbol{a}$ for $\rho_k$. By dividing the outcomes $\boldsymbol{a} = a_1 a_2 \cdots a_L$ into the first round $a_1 \in \{0, 1\}$ with $\Pi_{a_1}^{P_1}$ following distribution ${\cal A}_1$ and the remaining rounds $\bar{\boldsymbol{a}} = a_2 a_3 \cdots a_L$ following distribution ${\cal \bar{A}}$, we obtain
\begin{equation}
	I({\cal K} : {\cal A} ) = I({\cal K} : {\cal A}_1 {\cal \bar{A}} ) = I({\cal K}:{\cal A}_1) + I({\cal K}:{\cal \bar{A}} |{\cal A}_1),
\end{equation}
using the chain rule of mutual information in terms of the conditional mutual information $I({\cal K}:{\cal \bar{A}} |{\cal A}_1)$.

Holevo's theorem~\cite{holevo1974remarks} then bounds the conditional mutual information in terms of the entropy of the post-measurement states after the first measurement, $\rho_k^{P_1,a_1} =  \frac{\Pi_{a_1}^{P_1} \rho_k \Pi_{a_1}^{P_1}}{p(a_1 | k)}$ with $p(a_1 | k) = {\rm Tr} [\Pi_{a_1}^{P_1} \rho_k ] $, as
\begin{equation} 
I({\cal K}:{\cal \bar{A}} |{\cal A}_1) \leq \sum_{a_1} p(a_1) \, \chi\left( \left\{ p(k | a_1), \rho_k^{P_1,a_1} \right\} \right),
\end{equation}
where $\chi(\{ p_i, \rho_i \}) := S\left(\sum_i p_i \rho_i\right) - \sum_i p_i S(\rho_i)$ is the Holevo information with the von Neumann entropy $S(\rho) := - {\rm Tr} (\rho \log_2 \rho)$ and $p(k | a_1) = \frac{p(k, a_1)}{p(a_1)} = \frac{p_k p(a_1 | k)}{p(a_1)}$ is the distribution of $k$ conditioned on $a_1$. Assuming the uniform prior distribution, $p_k = \frac{1}{K}$, one then obtains the following bound on the conditional entropy between the labels $k$ and their guess $\boldsymbol{a}$ as 
\begin{equation}\label{eq:entropy_gap}
\begin{aligned}
H({\cal K} |{\cal A}) &= H({\cal K}) -  I({\cal K} : {\cal A})\\
& \geq \log_2 K - \left[ I({\cal K}:{\cal A}_1) + \sum_{a_1} p(a_1) \, \chi\left( \left\{ p(k | a_1), \rho_k^{P_1,a_1} \right\} \right) \right] =: \delta({\cal E}; P_1).
\end{aligned}
\end{equation}
The non-vanishing gap of the conditional entropy $H({\cal K} | {\cal A} ) > 0$ implies that perfect discrimination is not possible using any stabilizer operations. This gap remains valid even when the subsequent process allows nonstabilizer operations after the first-round stabilizer measurements. 

The success probability $p$ for the correct guess of $k$ from $\boldsymbol{a}$, or equivalently the error probability $e = 1 - p$ of discriminating the $K$ labels from the outcomes $\boldsymbol{a}$ is bounded by the conditional entropy using Fano's inequality~\cite{fano1961transmission}, $H({\cal K}  | {\cal A}) \leq H_b(e) + e\log_2 ( K - 1)$, where $H_b(x) = -x \log_2(x) - (1-x)\log_2(1-x) = H_b(1-x)$ for $0 \leq x \leq 1$ is the binary entropy. This leads to the SO-discrimination bound
$$
\min_{P_1} \delta({\cal E}; P_1) \leq H({\cal K}  | {\cal A}) \leq H_b(1-p^{\rm SO}_{\rm succ}({\cal E})) + (1- p^{\rm SO}_{\rm succ}({\cal E}))\log_2 ( K - 1).
$$

We explicitly evaluate this bound for the ensembles considered in the manuscript by considering all possible first-round Pauli measurements on the corresponding two- and three-qubit systems, which are summarized in table~\ref{tab:entropy_bounds}. Although these bounds are not tight, the entropic analysis provides an intuitive, information-theoretic demonstration that perfect discrimination is impossible without requiring an exhaustive search over all SO strategies. It is also straightforward to generalize this first-round measurement bound to involve more measurement outcomes, e.g., $a_1 \cdots a_l$, which will provide a tighter bound at the cost of a larger search over sequential Pauli-measurements.
\begin{table}[t]
\caption{Entropic bounds of the SO-indistinguishable stabilizer states.}
\centering
\begin{tabular}{ l | c|  c| c | c | c }
\hline
$\E$ & $n$ & $K$ & Conditional entropy bound & $p_{\rm succ}^{\rm SO}$ bound &  Optimal $p_{\rm succ}^{\rm SO}$ \\
\hline \hline
${\cal E}_{\rm 3} $ & 3 & 8 & $\geq 1/2 $  & $\leq 0.939$  & 3/4\\
${\cal E}_{\rm 3, SEP} $ & 3 & 6 & $\geq 1/3$  & $<0.960$ & 5/6\\ 
${\cal E}_{\rm 3, min} $ & 3 & 4 &  $\geq 0.094$  & $\leq 0.990$ & 15/16\\
${\cal E}_{\rm 3}^{\rm bin} $ & 3 & 2 &  $\geq 0.406$ & $\leq 0.919$ & 3/4\\
${\cal E}_{\rm 3, SEP}^{\rm bin} $ & 3 & 2 & $\geq 0.208$  & $\leq 0.967$ & 5/6\\ 
${\cal E}_{\rm 2} $  & 2 & 2 & $\geq 0.316$  & $\leq 0.943$ & 7/8\\
\hline
\end{tabular}
\label{tab:entropy_bounds}
\end{table}

\section{Formal definition of Pauli-decision trees}
\label{app:Pauli_decision_tree}

For a Hermitian Pauli operator $P\in{\cal P}_n$, recall that the Pauli-projection operator is defined as
\begin{equation}
    \Pi_a^P :=\frac{\mathbb{1}+(-1)^aP}{2},
\end{equation}
for outcomes $a\in\{0,1\}$. We introduce a Pauli-decision tree ${\cal T}$, a finite rooted binary tree in which an internal node with outcome history $\boldsymbol a_{<j}=(a_1,\ldots,a_{j-1})$ is assigned a Hermitian Pauli operator $P(\boldsymbol a_{<j})$ for the next measurement $\Pi_{a_j}^{P(\boldsymbol a_{<j})}$. Its two children correspond to the outcomes $a_j\in\{0,1\}$. A terminal history $\boldsymbol a=(a_1,\ldots,a_r)$ is called a leaf, and the set of leaves is denoted by ${\cal L}({\cal T})$. The root corresponds to the empty history $\boldsymbol a_{<1}=\emptyset$. For example, in the optimal Pauli-decision tree for discriminating ${\cal E}_3$ shown in figure~\ref{fig:states}(b), the history $\boldsymbol a_{<3}=(0,0)$ assigns the Pauli operator $P(\boldsymbol a_{<3})=P((0,0))=Z_2$. Outcome $a_3=0$ then leads to the leaf $\boldsymbol a=(0,0,0)$. Since every path in this example has length three, its leaf set is ${\cal L}({\cal T})=\{0,1\}^3$.

By the normal form of stabilizer operations \cite{heimendahl2022axiomatic}, after resolving input-independent classical randomness and absorbing intermediate Clifford operations into subsequent Pauli observables, it suffices to optimize over deterministic reduced Pauli-decision trees. Along every root-to-leaf path of a reduced tree, the measured Pauli observables are mutually commuting and independent modulo phases. A commuting Pauli observable generated, up to sign, by the preceding observables has a deterministic outcome and can be removed. Noncommuting measurements are eliminated through the normal-form reduction into input-independent classical randomness and conditional Clifford operations. Since the success probability is linear in the measurement, such classical randomization cannot improve the optimum.

For a leaf $\boldsymbol a=(a_1,\ldots,a_r)$ of a reduced tree, the corresponding leaf projector is
\begin{equation}
    \Pi_{\boldsymbol a} := \prod_{j=1}^{r} \Pi_{a_j}^{P(\boldsymbol a_{<j})}.
    \label{eq:leaf_projector_app}
\end{equation}
The ordering in this product is irrelevant because the Pauli observables commute along the path. Its signed stabilizer group is
\begin{equation}
    {\cal S}(\Pi_{\boldsymbol a}) := \left\langle (-1)^{a_j}P(\boldsymbol a_{<j}) : j=1,\ldots,r \right\rangle.
\end{equation}
Since the $r$ generators are independent, we can alternatively express the leaf projector as 
\begin{equation}
    \Pi_{\boldsymbol a} = \frac{1}{|{\cal S}(\Pi_{\boldsymbol a})|} \sum_{P \in{\cal S}(\Pi_{\boldsymbol a})}P,
    \label{eq:leaf_projector_rank}
\end{equation}
where the rank of the projector is given as $\operatorname{rank}(\Pi_{\boldsymbol a}) =2^{n-r}$. We note that leaves at different depths may have different ranks. The leaf projectors form a projective measurement, satisfying  $\Pi_{\boldsymbol a}\Pi_{\boldsymbol b} = \delta_{\boldsymbol a,\boldsymbol b}\Pi_{\boldsymbol a}$ and $ \sum_{\boldsymbol a\in{\cal L}({\cal T})} \Pi_{\boldsymbol a} = \mathbb{1}$. It is worth noting that the two children of each internal node resolve their parent projector, while two distinct leaves first differ at complementary outcomes of the same Pauli measurement. However, Pauli observables assigned to different branches need not commute.

Now we consider SO-discrimination of a given ensemble ${\cal E}=\{(p_k,\rho_k)\}_{k=1}^K$. The optimal guess associated with a leaf $\boldsymbol a$ is
\begin{equation}
    g(\boldsymbol a) = \arg\max_{1\leq k\leq K} p_k\Tr(\Pi_{\boldsymbol a}\rho_k).
\end{equation}
The success probability of a given Pauli-decision tree is therefore
\begin{equation}
    p_{\rm succ}^{\cal T}({\cal E}) = \sum_{k=1}^K p_k\sum_{\boldsymbol a\in{\cal L}({\cal T})} \Tr(\Pi_{\boldsymbol a}\rho_k) \delta_{k,g(\boldsymbol a)}
    = \sum_{\boldsymbol a\in{\cal L}({\cal T})} \max_{1\leq k\leq K} p_k\Tr(\Pi_{\boldsymbol a}\rho_k)
     = \sum_{\boldsymbol a\in{\cal L}({\cal T})} q(\Pi_{\boldsymbol{a}}),
\label{eq:tree_success_appendix}
\end{equation}
where we denote optimal unnormalized success contribution as
\begin{equation}
    q(\Pi) :=\max_{1\leq k\leq K} p_k\Tr(\Pi\rho_k).
\end{equation}
Every remaining measurement adds one independent commuting Pauli generator. Consequently, every root-to-leaf path has length at most $n$.

\section{Dynamic programming over stabilizer-code projectors}
\label{app:dynamic_programming}
We now consider the SO discrimination probability as the following optimization problem over Pauli-decision trees (see also Eq~\eqref{eq:SO_tree_optimization}),
\begin{equation}
    p_{\rm succ}^{\rm SO}({\cal E})
    = \max_{\cal T} \sum_{\boldsymbol a\in{\cal L}({\cal T})} \max_{1\leq k\leq K} p_k\Tr(\Pi_{\boldsymbol a}\rho_k)
    = \max_{\cal T} \sum_{\boldsymbol a\in{\cal L}({\cal T})} q(\Pi_{\boldsymbol{a}}).
    \label{eq:SO_tree_optimization_appendix}
\end{equation}
We evaluate this optimization by dynamic programming based on Bellman's principle of optimality~\cite{Bellman1952}, taking a stabilizer-code projector $\Pi$ as the dynamic-programming state. If ${\cal S}(\Pi)$ has $r$ independent generators, then $\operatorname{rank}(\Pi)=2^{n-r}$. We compare terminating the strategy at $\Pi$ with refining it by an additional Pauli measurement.

In particular, for a given stabilizer-code projector $\Pi$, we define the set of admissible Pauli operators by
\begin{equation}
\operatorname{Adm}(\Pi) := \left\{ P\in{\cal P}_n: P^\dagger=P, [P,{\cal S}(\Pi)]=0, \pm P \notin {\cal S}(\Pi) \right\}.
\label{eq:admissible_Pauli}
\end{equation}
The commutation condition preserves the stabilizer-code structure, while $\pm P \notin {\cal S}(\Pi)$ excludes trivial refinements having a deterministic outcome on the support of $\Pi$.

For $P\in\operatorname{Adm}(\Pi)$, the refinement is given by adding projector $\Pi_a^P$ to $\Pi$, given as
\begin{equation}
    \Pi_a^{P\mid\Pi} := \Pi_a^P \Pi,
    \label{eq:Bellman_child_projectors}
\end{equation}
$a\in\{0,1\}$. We note that the refined projectors satisfy $\Pi_0^{P\mid\Pi}+\Pi_1^{P\mid\Pi}=\Pi$ with decreased rank $\operatorname{rank}(\Pi_a^{P\mid\Pi})= 2^{n-r-1}$.  If both child branches are terminated immediately, their total success contribution is 
\begin{equation}
    q(\Pi_0^{P\mid\Pi}) + q(\Pi_1^{P\mid\Pi}),
\end{equation}
where different guessing labels may be assigned to the two outcomes.

More generally, either child may be refined further. Let $V(\Pi)$ denote the maximum success contribution obtainable from a decision subtree rooted at $\Pi$. Every such subtree either terminates at $\Pi$, contributing $q(\Pi)$, or begins with some $P\in\operatorname{Adm}(\Pi)$, after which the two child subtrees can be optimized independently. Therefore, the corresponding dynamic programming~\cite{Bellman1952} can be written as
\begin{equation}
    V(\Pi) = \max\left\{ q(\Pi), \max_{P\in\operatorname{Adm}(\Pi)} \left[ V \left(\Pi_0^{P\mid\Pi}\right) + V \left(\Pi_1^{P\mid\Pi}\right) \right] \right\},
    \label{eq:SO_Bellman_recursion_appendix}
\end{equation}
where the second maximization is omitted when $\operatorname{Adm}(\Pi)=\varnothing$.

This recursion is exact because every subtree rooted at $\Pi$ has precisely one of the two forms appearing on the right-hand side. Conversely, attaching optimal subtrees to the two outcomes of any admissible Pauli measurement produces a valid subtree rooted at $\Pi$. Moreover, every admissible refinement adds one independent stabilizer generator, so the recursion reaches a rank-one projector after at most $n-r$ further measurements. For a rank-one projector $\Pi$, the recursion terminates as $\operatorname{Adm}(\Pi)=\varnothing$ and $V(\Pi)=q(\Pi)$. Backward induction therefore establishes the optimality of the recursion. Starting from the root gives
\begin{equation}
	p_{\rm succ}^{\rm SO}({\cal E})=V(\mathbb{1}).
\label{eq:Bellman_root_value}
\end{equation}
The corresponding pseudocode is given in Algorithm~\ref{alg:SO_dynamic_programming}.

\begin{algorithm}[!ht]
\caption{Exact SO discrimination by memoized dynamic programming}
\label{alg:SO_dynamic_programming}
\DontPrintSemicolon
\SetKwProg{Fn}{Function}{:}{end}
\SetKwFor{ForEach}{for each}{do}{end for}
\SetKw{InSet}{in}

\KwIn{An ensemble ${\cal E}=\{(p_k,\rho_k)\}_{k=1}^K$}
\KwOut{The optimal success probability
$p_{\rm succ}^{\rm SO}({\cal E})$}

$\mathsf{Memo}\gets\varnothing$
\BlankLine

\Fn{$V(\Pi)$}{
    \If{$\Pi$ is a key in $\mathsf{Memo}$}{
        \KwRet{$\mathsf{Memo}[\Pi]$} }
    $v\gets q(\Pi)$\;

\ForEach{$P$ \InSet $\operatorname{Adm}(\Pi)$}{$\Pi_0\gets\Pi_0^P\Pi$\;
        $\Pi_1\gets\Pi_1^P\Pi$\;
        $v\gets\max \left\{
        v,V(\Pi_0)+V(\Pi_1)
        \right\}$\;
    }

    $\mathsf{Memo}[\Pi]\gets v$\;
    \KwRet{$v$}\;
}

\KwRet{$V(\mathbb{1})$}\;
\end{algorithm}

It is worth noting that, when the projectors are stored using a canonical stabilizer representation, memoization ensures that every encountered stabilizer-code projector is evaluated at most once. If a maximizing action is also recorded at each projector, an optimal Pauli-decision tree can be reconstructed.

Although $\operatorname{Adm}(\Pi)$ as defined above may contain redundant Pauli operators, these redundancies may be removed in the implementation. In particular, $P$ and $P'=\pm PQ$ with $Q\in{\cal S}(\Pi)$ generate the same two child projectors, possibly with their outcome labels exchanged. For a projector with $r$ independent stabilizer constraints, the number of distinct admissible refinements after removing these redundancies is $4^{n-r}-1$. The number of stabilizer-code projectors with $r$ independent constraints is
\begin{equation}
N_{n,r} = 2^r\prod_{j=0}^{r-1} \frac{2^{2(n-j)}-1}{2^{r-j}-1}.
\end{equation}
Consequently, the total number of distinct refinement candidates is
\begin{equation}
\sum_{r=0}^{n-1}
N_{n,r}\left(4^{n-r}-1\right) = 2^{O(n^2)}.
\end{equation}

\section{SO indistinguishability of ${\cal E}_{n, {\rm SEP}}$}\label{app:n-qubit}
We show that a set of SO-indistinguishable stabilizer states in Theorem~\ref{thm:main} can be generalized to $n = n_a + n_b + n_c$ qubits by considering the following set of mutually orthogonal states:
\begin{equation}\label{eq_app:n-qubit_states}
\begin{aligned}
	\ket{\psi_A(\boldsymbol{x}_a, \boldsymbol{z}_b^1, \boldsymbol{z}_c^0)} &= \ket{\boldsymbol{x}_a} \otimes \ket{\boldsymbol{z}_b^{1}} \otimes \ket{\boldsymbol{z}_c^{0}} \\
	\ket{\psi_B(\boldsymbol{z}_a^0, \boldsymbol{x}_b, \boldsymbol{z}_c^1)} &= \ket{\boldsymbol{z}_a^0} \otimes \ket{\boldsymbol{x}_b} \otimes \ket{\boldsymbol{z}_c^{1}} \\
	\ket{\psi_C(\boldsymbol{z}_a^1, \boldsymbol{z}_b^0, \boldsymbol{x}_c)} &= \ket{\boldsymbol{z}_a^1} \otimes \ket{\boldsymbol{z}_b^{0}} \otimes \ket{\boldsymbol{x}_c},
	\end{aligned}
\end{equation}
where  $\boldsymbol{x}_{a,b,c} \in \{ +, -\}^{ n_{a,b,c}}$ and $\boldsymbol{z}_{a,b,c}^{0,1} \in {\cal Z}_{a,b,c}^{0,1} \subset \{ 0, 1\}^{n_{a,b,c}}$ such that $f_{a,b,c} (\boldsymbol{z}_{a,b,c}^0) = 0$ and $f_{a,b,c} (\boldsymbol{z}_{a,b,c}^1) = 1$ for Boolean functions $f_{a,b,c}(\boldsymbol{z}_{a,b,c}) : \{0,1\}^{n_{a,b,c}} \to \{0,1\}$ with vanishing linear structure. 

The proof proceeds similarly to Theorem~\ref{thm:main}. As Lemma~\ref{lemma:stab_recover} also applies here, we only need to consider the ancilla-free case. First, note that the union of all stabilizers in the set contains every $X_j$ and $Z_j$ for $j = 1, 2, \dots, n$, as in the $3$-qubit case. Hence, every nonidentity Pauli observable $P$ anticommutes with at least one single-qubit $X_j$ or $Z_j$. We then consider two different cases: (i) $P Z_j = -Z_j P$ for some $j \in \{1, \dots, n\}$, and (ii) $P$ commutes with every Pauli $Z$'s, implying that $P X_j = -X_j P$ for some $j \in \{1, \dots, n\}$.
\begin{enumerate}
\item[(i)] If $P Z_j = - Z_j P$ for some $j \in \{ 1, \dots, n\}$, then at least one of the following sets, ${\cal Z}_A = \{ Z_1, \dots, Z_{n_a}\}$, ${\cal Z}_B = \{ Z_{n_a + 1}, \dots, Z_{n_a+n_b}\}$, or ${\cal Z}_C = \{ Z_{n_a + n_b + 1}, \dots, Z_{n}\}$, contains an element that anticommutes with $P$. Let us first suppose that $Z_j \in {\cal Z}_A$ anticommutes with $P$. Define an $n_a$-bit string $\boldsymbol\gamma_a = \gamma_1 \dots \gamma_{n_a} \in \{ 0,1\}^{n_a}$ by setting $\gamma_i = 0$ if $Z_i$ commutes with $P$ and $\gamma_i = 1$ if $Z_i$ anticommutes with $P$. Note that $\boldsymbol\gamma_a \neq 0$. By the vanishing linear structure of $f_a$, for any non-zero $\boldsymbol\gamma_a$, there exists $\boldsymbol{z}_a^0$ such that $f_a(\boldsymbol{z}_a^0) = 0$ and $f_a(\boldsymbol{z}_a^0 \oplus \boldsymbol\gamma_a) = 1$. Consider $\ket{\psi_B(\boldsymbol{z}_a^0, \boldsymbol{x}_b, \boldsymbol{z}_c^1)} = \ket{\boldsymbol{z}_a^0} \ket{\boldsymbol{x}_b} \ket{\boldsymbol{z}_c^{1}}$. Applying $P$ gives
\begin{equation}
P \ket{\psi_B(\boldsymbol{z}_a^0, \boldsymbol{x}_b, \boldsymbol{z}_c^1)} = P \ket{\boldsymbol{z}_a^0} \ket{\boldsymbol{x}_b} \ket{\boldsymbol{z}_c^{1}} = \omega \ket{\boldsymbol{z}_a^0 \oplus \boldsymbol\gamma_a} \ket{\boldsymbol{x}'_b} \ket{\boldsymbol{z}'_c},
\end{equation}
for some $\omega \in \{\pm 1, \pm i\}$, $\boldsymbol{x}'_b \in \{ +, -\}^{n_b}$, and $\boldsymbol{z}'_c \in \{ 0, 1 \}^{n_c}$. Since $f(\boldsymbol{z}_a^0 \oplus \boldsymbol\gamma_a) = 1$, by taking $\boldsymbol{z}_a^1 = \boldsymbol{z}_a^0 \oplus \boldsymbol\gamma_a$, there exists  $\ket{\psi_C(\boldsymbol{z}_a^1, \boldsymbol{z}_b^0, \boldsymbol{x}_c)} = \ket{\boldsymbol{z}_a^1} \ket{\boldsymbol{z}_b^{0}} \ket{\boldsymbol{x}_c}$ such that 
\begin{equation}
\begin{aligned}
	\bra{\psi_B(\boldsymbol{z}_a^0, \boldsymbol{x}_b, \boldsymbol{z}_c^1)} P \ket{\psi_B(\boldsymbol{z}_a^0, \boldsymbol{x}_b, \boldsymbol{z}_c^1)} & = 0\\
	\left| \bra{\psi_C(\boldsymbol{z}_a^1, \boldsymbol{z}_b^0, \boldsymbol{x}_c)} P \ket{\psi_B(\boldsymbol{z}_a^0, \boldsymbol{x}_b, \boldsymbol{z}_c^1)} \right|
	&= \left| \langle \boldsymbol{z}_a^1 | \boldsymbol{z}_a^0 \oplus \boldsymbol\gamma_a \rangle \langle \boldsymbol{z}_b^0 | \boldsymbol{x}'_b \rangle \langle \boldsymbol{x}_c | \boldsymbol{z}'_c \rangle \right | 
	= 2^{-\left( \frac{n_b + n_c}{2} \right)} > 0,
\end{aligned}
\end{equation}
as the overlap between $\ket{ \boldsymbol{x}_b}$ ($\ket{ \boldsymbol{x}_c}$) and $\ket{ \boldsymbol{z}_{b}}$ ($\ket{ \boldsymbol{z}_{c}}$) is always $2^{-{n_b}/2}$ ($2^{-{n_c}/2}$) up to a phase. For the other cases $Z_j \in {\cal Z}_B$ and $Z_j \in {\cal Z}_C$, the same logic applies. 

\item[(ii)] If $P$ commutes with every Pauli-$Z$'s and $P X_j = - X_j P$ for some $j$, similarly, at least one of the sets ${\cal X}_A = \{X_1,\dots,X_{n_a}\}$, ${\cal X}_B = \{X_{n_a+1},\dots,X_{n_a + n_b}\}$, or ${\cal X}_C = \{X_{n_a + n_b+1},\dots,X_{n}\}$ contains an element that anticommutes with $P$. Let us first assume that $X_j \in {\cal X}_A$ anticommute with $P$. We then define $\boldsymbol\delta = \delta_1 \dots \delta_{n_a} \in \{ 0, 1 \}^{n_a}$ with $\delta_i = 0$ if $X_i$ commutes with $P$ and $\delta_i = 1$ if $X_i$ anticommutes with $P$. By taking $\ket{\psi_A(\boldsymbol{x}_a, \boldsymbol{z}_b^1, \boldsymbol{z}_c^0)} = \omega \ket{\boldsymbol{x}_a}  \ket{\boldsymbol{z}_b^{1}} \ket{\boldsymbol{z}_c^{0}}$ (up to some phase $\omega \in \{ \pm 1, \pm i \}$) stabilized by $X_j \in {\cal X}_A$, there always exists another state $\ket{\psi_A(\boldsymbol{x}'_a, \boldsymbol{z}_b^1, \boldsymbol{z}_c^0)}$ (as $\boldsymbol{\delta} \neq 0$), where $(\boldsymbol{x}'_a)_i = (-1)^{\delta_i} (\boldsymbol{x}_a)_i$ for all $i = 1, \cdots, n_a$,  such that $P \ket{\psi_A(\boldsymbol{x}_a, \boldsymbol{z}_b^1, \boldsymbol{z}_c^0)}$ perfectly overlaps with $\ket{\psi_A(\boldsymbol{x}'_a, \boldsymbol{z}_b^1, \boldsymbol{z}_c^0)}$, as all remaining $Z$ stabilizers are unchanged and the $X$ stabilizers match. For the other cases $X_j \in {\cal X}_B$ and $X_j \in {\cal X}_C$, the same logic applies.
\end{enumerate}

Hence, for any $P \neq \mathbb{1}$, there exists a pair of quantum states $\ket{\psi_k}, \ket{\psi_{k'}} \in {\cal E}_{n, {\rm SEP}}$ such that $\bra{\psi_k} P \ket{\psi_k} = 0$ and $|\bra{\psi_{k'}} P \ket{\psi_k}|^2 > 0$ as in Theorem~\ref{thm:main}, which completes the proof.

\section{Explicit adaptive discrimination protocol with a single Toffoli gate}
\label{app:3q_protocol}
We present an adaptive protocol that perfectly discriminates both ${\cal E}_{\rm 3,SEP}$ and ${\cal E}_{\rm 3}$ using a single Toffoli gate (i.e., the $\rm CCX$ gate). Let us recall the states
\begin{equation}\label{eq_app:3q-states}
\begin{aligned}
	\ket{\psi_A^{\pm}} & = \frac{\ket{010} \pm \ket{110}}{\sqrt{2}} \\
	\ket{\psi_B^{\pm}} & = \frac{\ket{001} \pm \ket{011}}{\sqrt{2}} \\
	\ket{\psi_C^{\pm}} & = \frac{\ket{100} \pm \ket{101}}{\sqrt{2}} \\
	\ket{\psi_{\rm GHZ}^{\pm}} & = \frac{\ket{000} \pm \ket{111}}{\sqrt{2}}.
\end{aligned}
\end{equation}
We note that the states in different blocks $A, B, C, {\rm GHZ}$ are spanned by distinct computational basis states. Now, let us consider the Boolean function $f(z_1z_2z_3): \{ 0,1 \}^3 \to \{ 0, 1\}$ such that 
\begin{equation}\label{eq:Toffoli-Boolean-partition}
\begin{aligned}
    f(000) = f(010) = f(110) = f(111) &= 0 \\
    f(001) = f(011) = f(100) = f(101) &= 1.
\end{aligned}
\end{equation}
We find that the following Boolean function
\begin{equation}
    f(z_1z_2z_3) := z_1 \oplus z_3 \oplus z_1 z_2 \oplus z_1 z_3 = z_1 \oplus z_3 \oplus z_1 ( z_2 \oplus z_3)
    \label{eq:Toffoli-Boolean-function}
\end{equation}
leads to Eq.~\eqref{eq:Toffoli-Boolean-partition}. Based on this, by appending a register state $\ket{0}_r$, we construct a unitary operation
\begin{equation}
    U_f\ket{z_1 z_2 z_3 }\ket{y}_r = \ket{z_1 z_2 z_3}\ket{y \oplus f(z_1 z_2 z_3)}_r.
    \label{eq:Toffoli-predicate-action}
\end{equation}
As coherent Boolean multiplication, e.g., $\ket{z_1} \ket{z_2} \ket{y} \rightarrow \ket{z_1} \ket{z_2} \ket{y + z_1 z_2}$, cannot be implemented by Clifford unitaries, one requires a non-Clifford gate, such as the Toffoli gate, to implement such a unitary. In this case, single multiplication $z_1 ( z_2 \oplus z_3)$ required for $U_f$ can be implemented by a single Toffoli gate, while an explicit gate sequence could be written as 
\begin{equation}
    U_f := \operatorname{CNOT}_{3\rightarrow 2} \operatorname{CCX}_{1,2\rightarrow r} \operatorname{CNOT}_{3\rightarrow 2} \operatorname{CNOT}_{3\rightarrow r} \operatorname{CNOT}_{1\rightarrow r},
\label{eq:Toffoli-predicate-circuit}
\end{equation}
where $a \rightarrow b$ denotes the control qubit $a$ and the target qubit $b$ and $\operatorname{CCX}$ denotes the Toffoli gate. 

Consequently, based on the register state measurement in the computational basis, one can perfectly identify whether the target state is in $\left\{\psi_A^{\pm}, \psi_{\rm GHZ}^{\pm} \right\}$ (for outcome $0$) or $\left\{ \psi_B^{\pm}, \psi_C^{\pm} \right\}$ (for outcome $1$). We now show that for each case, SO can fully discriminate the states in the corresponding subset.
\begin{enumerate}
	\item For outcome 0: $\left\{\psi_A^{\pm}, \psi_{\rm GHZ}^{\pm} \right\}$ can be perfectly distinguished by first performing $\Pi^{Z_2 Z_3}_a$ with $Z_2 Z_3 \ket{\psi_A^{\pm}} = - \ket{\psi_A^{\pm}}$ and $Z_2 Z_3 \ket{ \psi_{\rm GHZ}^{\pm}} = \ket{ \psi_{\rm GHZ}^{\pm}}$. After determining whether the state belongs to $A$ or ${\rm GHZ}$, the phase information could be fully identified by $\Pi^{X_1}_{s}$ or $\Pi^{X_1 X_2 X_3}_{s}$, respectively, by noting that $X_1 \ket{\psi_A^{\pm}} = \pm \ket{\psi_A^{\pm}}$, $X_1 X_2 X_3 \ket{ \psi_{\rm GHZ}^{\pm}} = \pm \ket{ \psi_{\rm GHZ}^{\pm}}$, and both $X_1$ and $X_1 X_2 X_3$ commute with $Z_2 Z_3$.
	\item For outcome 1: $\left\{ \psi_B^{\pm}, \psi_C^{\pm} \right\}$ can be perfectly distinguished by measuring the first qubit, i.e., $\Pi^{Z_1}_a$, which determines whether the state belongs to $B$ or $C$. Then, taking $\Pi^{X_2}_s$ (for $B$) or $\Pi^{X_3}_s$ (for $C$) could determine the sign, so the states can be perfectly discriminated by SO.
\end{enumerate}
Therefore, ${\cal E}_3$ can be fully discriminated by a single Toffoli gate and adaptive Pauli measurements (see the decision tree in figure~\ref{fig:E3_decision_tree}). The same protocol also perfectly discriminates ${\cal E}_{\rm 3, SEP}$, since ${\cal E}_{\rm 3, SEP} \subset {\cal E}_3$.
\begin{figure}[ht]
    \centering
    \includegraphics[width=0.5\linewidth]{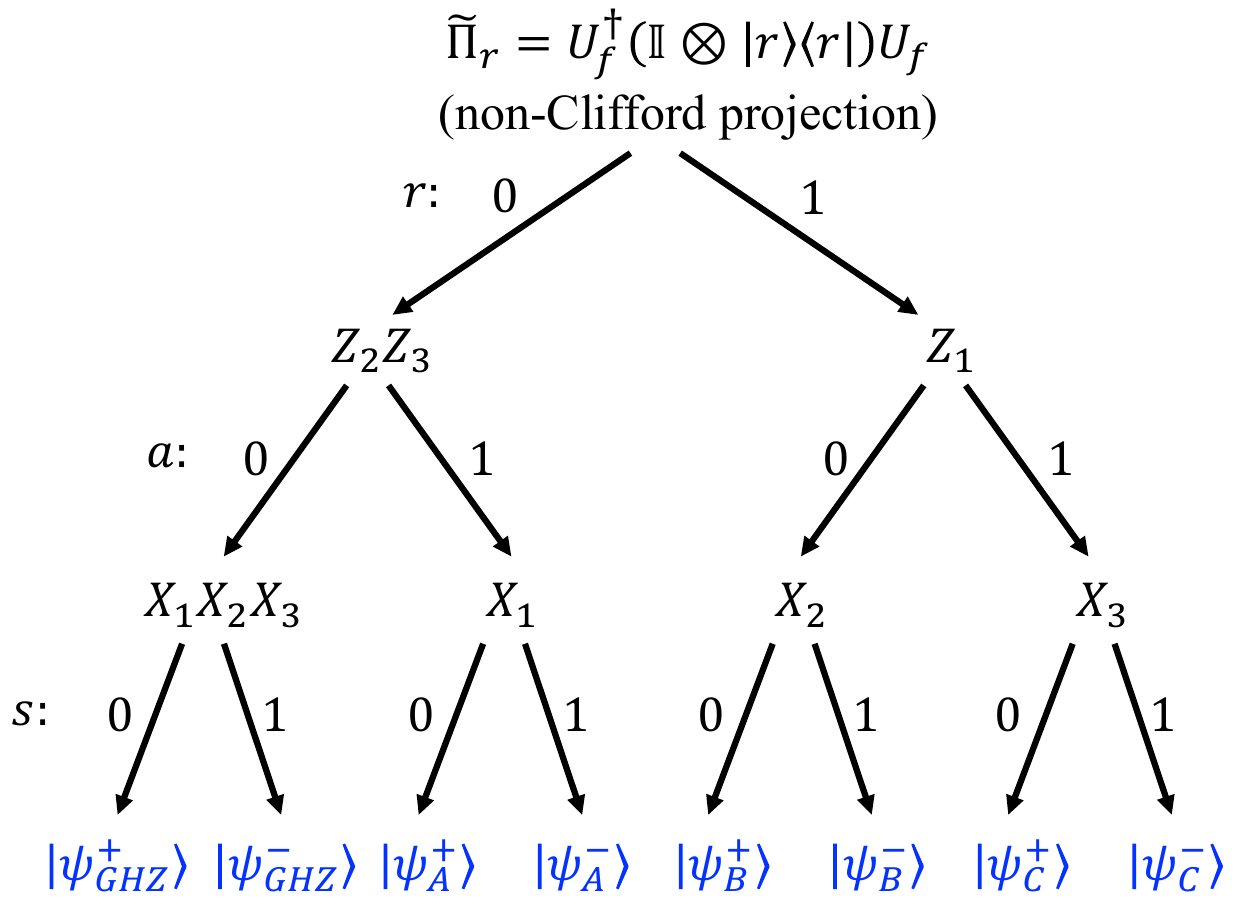}
    \caption{Decision tree for perfect discrimination of ${\cal E}_3$ and ${\cal E}_{\rm 3, SEP}$. Only the first projection $\tilde{\Pi}_r$ contains non-stabilizerness with a single Toffoli gate.}
    \label{fig:E3_decision_tree}
\end{figure}

\section{Impossibility of perfectly discriminating ${\cal E}_{\rm 2}$ with finitely many $T$ states}
\label{app:finite_T_E2}
For a binary ensemble ${\cal E}_{\rm 2} =\{\rho_0,\rho_1\}$ with
\begin{equation}
    \qquad
    \rho_0=
    \frac{\ket{0{+}}\bra{0{+}}+\ket{+0}\bra{+0}}{2},
    \qquad
    \rho_1=
    \frac{\ket{11}\bra{11}+\ket{-{-}}\bra{-{-}}}{2},
\end{equation}
we have the following statement:
\begin{theorem}
No finite SO protocol assisted by a finite number of $T$ states can perfectly discriminate ${\cal E}_{\rm 2}$. 
\end{theorem}

\begin{proof}
We first express
\begin{equation}
	\Delta = \rho_0 - \rho_1 = 
    \frac{1}{8}
    \begin{pmatrix}
        3&3&3&-1\\
        3&1&-1&1\\
        3&-1&1&1\\
        -1&1&1&-5
    \end{pmatrix}
\end{equation}
in the two-qubit computational basis $\{ \ket{00}, \ket{01}, \ket{10}, \ket{11} \}$. The spectral decomposition leads to the eigenvalues of
\begin{equation}
    \operatorname{eigs}(\Delta) = \left\{ \frac{3}{4}, \frac{1}{4}, -\frac{1}{4}, -\frac{3}{4} \right\},
\end{equation}
with the corresponding eigenstates
\begin{equation}
\begin{aligned}
\ket{v_{3/4}}
&=\frac{2\ket{00}+\ket{01}+\ket{10}}{\sqrt{6}},\qquad && \ket{v_{1/4}} =\frac{\ket{01}-\ket{10}}{\sqrt{2}},\\
\ket{v_{-1/4}}
&=\frac{\ket{00}  - \ket{01} - \ket{10} - \ket{11}}{2},\qquad && \ket{v_{-3/4}}=\frac{\ket{00}-\ket{01}-\ket{10}+3\ket{11}}
{2\sqrt{3}}.
\end{aligned}
\end{equation}
Since $\Delta$ is full rank, the Helstrom measurement associated with the guess $\rho_0$ is uniquely given by the projection onto the positive eigenstates, 
\begin{equation}
M_{\rm H} = \ket{v_{3/4}}\bra{v_{3/4}} +\ket{v_{1/4}}\bra{v_{1/4}} = \frac{1}{3}
\begin{pmatrix}
2&1&1&0\\
1&2&-1&0\\
1&-1&2&0\\
0&0&0&0
\end{pmatrix}.
\end{equation}

Now, let us consider the optimal SO-discrimination using $t$-copies of $T$ states, whose effective measurement operator becomes
\begin{equation}
\begin{aligned}
	p_{\rm succ}^{\rm SO}({\cal E}_2 \otimes \psi_T^{\otimes t}) &= \frac{1}{2} \left( 1 + \frac{\| \left( \Delta \otimes \ket{T}\bra{T}^{\otimes t} \right) \|_{\rm SO}}{2} \right)  \\
	&= \frac{1}{2} \left( 1 +  \Tr \left[ (\Delta  \otimes \ket{T}\bra{T}^{\otimes t} ) M \right]  \right)  \\
	&= \frac{1}{2} \left( 1 + \Tr \left[ \Delta  M^{(t)} \right]  \right),
	\end{aligned}
\end{equation}
where $M$ corresponds to the optimal SO-measurement operator, and we denote
\begin{equation}
    M^{(t)} = \bra{T}^{\otimes t} M \ket{T}^{\otimes t}.
    \label{eq:E2_effective_effect}
\end{equation}
If the perfect discrimination of ${\cal E}_2$ is possible using additional $t$-copies of $T$ states, the measurement effect should be $M^{(t)} = M_{\rm H}$. However, we show that this is not possible for any finite $t$, as 
\begin{equation}
	\bra{01} M^{(t)} \ket{01} \neq \frac{2}{3} = \bra{01} M_{\rm H} \ket{01}.
\end{equation}

This can be shown by following a similar logic from Ref.~\cite{giles2013exact} for exact Clifford+$T$ synthesis of unitary quantum circuits. We first consider a deterministic finite SO protocol. Using the reduced Pauli-decision-tree representation introduced in section~\ref{app:Pauli_decision_tree}, its joint effect associated with the guess $\rho_0$ can be written as
\begin{equation}
M = \sum_{\boldsymbol a\in{\cal L}_0({\cal T})} \Pi_{\boldsymbol a},
\end{equation}
where ${\cal L}_0({\cal T})$ denotes the leaves assigned the guess $\rho_0$. By noting that 
\begin{equation}
\ket{T}^{\otimes t} = \frac{1}{2^{t/2}} \sum_{\boldsymbol z\in\{0,1\}^t} \omega^{|\boldsymbol z|} \ket{\boldsymbol z},
\end{equation}
where $\omega=e^{i\pi/4} = \frac{1+ i}{\sqrt{2}}$ and $|\boldsymbol z|$ denotes the Hamming weight of $\boldsymbol z$, the overlap in Eq.~\eqref{eq:E2_effective_effect} can be written as
\begin{equation}
\bra{01}M^{(t)}\ket{01} = \frac{1}{2^t} \sum_{\boldsymbol a\in{\cal L}_0({\cal T})} \sum_{\boldsymbol z,\boldsymbol z'\in\{0,1\}^t} \omega^{|\boldsymbol z|-|\boldsymbol z'|} \bra{01\boldsymbol z'} \Pi_{\boldsymbol a} \ket{01\boldsymbol z}.
\label{eq:E2_effect_expansion}
\end{equation}
However, $\bra{01\boldsymbol z'} \Pi_{\boldsymbol a} \ket{01\boldsymbol z}$ on the entire $(2+t)$-qubit system can only take values of $(a' + i c') 2^{-(2+t) } \in \mathbb{Z}\left[ \frac{1}{2}, i \right]$ with $a', c' \in \mathbb{Z}$, while the phase $\omega^{|\boldsymbol z| - |\boldsymbol z'|}$ can only take values of $\mathbb{Z}\left[\frac{1}{\sqrt{2}},i\right]$. Hence, each term in the summation in Eq.~\eqref{eq:E2_effect_expansion} can be written as
\begin{equation}
\omega^{|\boldsymbol z|-|\boldsymbol z'|}  \bra{01\boldsymbol z'} \Pi_{\boldsymbol a} \ket{01\boldsymbol z}  = \frac{a+b\sqrt{2}+i(c+d\sqrt{2})}{2^m} \in \mathbb{Z}\left[\frac{1}{\sqrt{2}},i\right], \quad a,b,c,d\in\mathbb{Z}, \quad m\in\mathbb{N}_0.
\end{equation}
Consequently,
\begin{equation}
\bra{01}M^{(t)}\ket{01} \in \mathbb{Z}\left[\frac{1}{\sqrt{2}},i\right],
\end{equation}
which is the coefficient ring associated with exact Clifford+$T$ circuits~\cite{giles2013exact}. Finally, by noting that $\bra{01} M_{\rm H} \ket{01} = 2/3 \notin \mathbb{Z}\left[\frac{1}{\sqrt{2}},i\right]$, one cannot exactly implement $M_{\rm H}$ from any finite $t$-copies of ancillary $T$-states and SO. Allowing input-independent classical randomization over effects $\{(p_j,M_j^{(t)})\}_j$ does not change the conclusion. If the randomized protocol discriminates perfectly, every component with $p_j>0$ must also discriminate perfectly and hence, by uniqueness of the Helstrom effect, must equal $M_{\rm H}$. This completes the proof.
\end{proof}

Nevertheless, the success probability can approach unity as the number of $T$ states increases. By the approximate universality of Clifford+$T$ circuits~\cite{Kliuchnikov2013asymptotically}, projective measurements implementable by finite Clifford+$T$ circuits are dense in the set of all projective measurements. Since $M_{\rm H}$ is a projective effect, for every $\epsilon>0$ there exist a finite $t$ and an effective measurement effect $M^{(t)}$ such that
\begin{equation}
    \left\|M^{(t)}-M_{\rm H}\right\|_{\infty}\leq\epsilon.
\end{equation}
Each $T$ gate can be implemented by consuming one $T$ state and otherwise using only SO~\cite{zhou2000methodology}; hence this $M^{(t)}$ is realizable by an SO protocol assisted by finitely many $T$ states. The optimal assisted success probability therefore satisfies
\begin{equation}
\begin{aligned}
p_{\rm succ}^{\rm SO} \left( {\cal E}_{\rm 2}\otimes\psi_T^{\otimes t} \right)
&\geq \frac{1}{2}\left(1+\Tr[\Delta M^{(t)}]\right)\\
&= 1+\frac{1}{2}\Tr \left[ \Delta ( M^{(t)}-M_{\rm H} ) \right]\\
&\geq 1-\frac{1}{2}\|\Delta\|_1 \left\|M^{(t)}-M_{\rm H}\right\|_{\infty}\\
&\geq 1-\epsilon,
\end{aligned}
\end{equation}
where we used $\Tr(\Delta M_{\rm H})=1$ and $\|\Delta\|_1=2$. We also numerically evaluate that a single $T$ state, i.e., $t=1$ already leads to the success probability $p_{\rm succ}^{\rm SO}({\cal E}_2;\psi_T) = 0.9268$, and for $t=4$, $p_{\rm succ}^{\rm SO}({\cal E}_2;\psi_T^{\otimes 4}) = 0.9973$, which is close to $1$ (see figure~\ref{fig:magic_bound}(b) in the main text).

\section{Proof of Theorem~\ref{thm:robustness_bound}} \label{app:robustness}
\begin{proof}
Let us consider the optimal stabilizer state decomposition of the appended ancillary state,
\begin{equation} \label{eq_app:robust_decomp}
    \sigma = \sum_j q_j \ket{\phi_j}\bra{\phi_j}, \quad \forall \ket{\phi_j} \in\operatorname{STAB}~{\rm and}~\sum_j q_j = 1,
\end{equation}
which leads to the robustness measure in Eq.~\eqref{eq:robustness_magic} with ${\cal R}(\sigma) = \sum_j | q_j |$. Now let us consider the optimal SO measurement ${\cal M}_{\rm SO}= \{ M, \mathbb{1}- M\} $ on the joint system $\rho_k \otimes \sigma$ to perform a discrimination of an ensemble ${\cal E}^{\rm bin} = \{ \rho_0, \rho_1\}$. Then, from Eqs.~\eqref{eq:bin_disc} and \eqref{eq:SO_norm}, we have the $\sigma$-assisted SO discrimination probability
\begin{equation}
    p_{\rm succ}^{\rm SO}({\cal E}^{\rm bin} \otimes \sigma) = \frac{1}{2} \left( 1 + \frac{\| \Delta \otimes \sigma \|_{\rm SO}}{2} \right) = \frac{1}{2} \left( 1 + \left|\Tr[ (\Delta \otimes \sigma) M ] \right| \right),
\end{equation}
where $\Delta = \rho_0 - \rho_1$. Using the decomposition in Eq.~\eqref{eq_app:robust_decomp}, we have
\begin{equation}
\begin{aligned}
	\left|\Tr[ (\Delta \otimes \sigma) M ] \right| &=  \left| \sum_j q_j \Tr[ (\Delta \otimes \ket{\phi_j} \bra{\phi_j}) M ] \right| \\
	&\leq \sum_j |q_j| \left|  \Tr[ (\Delta \otimes \ket{\phi_j} \bra{\phi_j}) M ] \right| \\
	&= \sum_j |q_j| \left|  \Tr[ \Delta M^{(j)} ] \right| \\
	&\leq \left( \sum_j |q_j| \right) \left( 2 p_{\rm succ}^{\rm SO}({\cal E}^{\rm bin}) -1 \right) \\
	&= {\cal R}(\sigma) \left( 2 p_{\rm succ}^{\rm SO}({\cal E}^{\rm bin}) -1 \right),
\end{aligned}
\end{equation}
where $M^{(j)} := \bra{\phi_j} M \ket{\phi_j}$ could be understood as a measurement operator acting on $\rho_k$, which is realized by appending a stabilizer state $\ket{\phi_j} \bra{\phi_j}$ followed by the joint system measurement with ${\cal M}_{\rm SO} = \{ M, \mathbb{1}-M\}$. As all these steps are SO, the effective POVM $\{ M^{(j)}, \mathbb{1} - M^{(j)}\}$ is SO-implementable, from which the second inequality follows. By arranging the expression, we obtain the desired inequality in Eq.~\eqref{eq:robustness_success_bound_bin}.

For a general case of ${\cal E} = \{ (p_k, \rho_k) \}$, we also obtain a similar bound
\begin{equation}
\begin{aligned}
    p_{\rm succ}^{\rm SO}({\cal E} \otimes \sigma) &= \sum_{k=1}^K p_k\Tr[ (\rho_k \otimes \sigma) M_k] \\
    &= \sum_j q_j \sum_{k=1}^K p_k  \Tr[ (\rho_k \otimes \ket{\phi_j} \bra{\phi_j}) M_k] \\
    &\leq \sum_j |q_j| \sum_{k=1}^K p_k  \Tr[ \rho_k M_k^{(j)}] \\
    &\leq {\cal R}(\sigma) p_{\rm succ}^{\rm SO} ({\cal E}),
\end{aligned}
\end{equation}
where $M_k^{(j)} = \bra{\phi_j} M_k \ket{\phi_j}$ is the measurement operator with ${\cal M}_{\rm SO} = \{ M_k \}$ being the optimal SO-measurement on the joint system for discriminating ${\cal E} \otimes \sigma := \{ p_k, \rho_k \otimes \sigma \}_{k=1}^K$. Similarly to the binary case, $\{ M_k^{(j)} \}_{k=1}^K$ is an SO-implementable POVM, so that the last inequality holds. This proves Eq.~\eqref{eq:robustness_success_bound}, which completes the proof. We note that for $K=2$, Eq.~\eqref{eq:robustness_success_bound_bin} implies Eq.~\eqref{eq:robustness_success_bound} as
\begin{equation}
	    p_{\rm succ}^{\rm SO} ({\cal E}^{\rm bin};\sigma) \leq \frac{1}{2} + {\cal R}(\sigma)  \left( p_{\rm succ}^{\rm SO} ({\cal E}^{\rm bin}) - \frac{1}{2} \right) = {\cal R}(\sigma)  p_{\rm succ}^{\rm SO} ({\cal E}^{\rm bin}) - \frac{1}{2} ({\cal R}(\sigma) - 1) \leq {\cal R}(\sigma)  p_{\rm succ}^{\rm SO} ({\cal E}^{\rm bin}),
\end{equation}
since ${\cal R}(\sigma) \geq 1$.
\end{proof}

\section{Proof of Theorem~\ref{prop:SO_LOCC}}
\label{app:stab_LOCC_prop}
\begin{proof}
We prove the contrapositive by showing that if ${\cal E}=\{ \psi_k \}_{k=1}^K$, consisting of products of single-qubit stabilizer states
\begin{equation}
    \ket{\psi_k} = \bigotimes_{j=1}^n\ket{\phi_{k,j}},~\forall \ket{\phi_{k,j}}\in{\rm STAB}_1,
\end{equation}
is LOCC-distinguishable within finite rounds, then ${\cal E}$ is SO-distinguishable, more precisely, by adaptive local Pauli measurements. Here, each qubit is held by a separate party, and LOCC refers to finite-round LOCC. To this end, we use the result from Refs.~\cite{Walgate2002nonlocality, Duan2010locally} that if multiqubit orthogonal states are perfectly distinguishable by finite-round LOCC, then they can be perfectly distinguished by adaptive local rank-one projective measurements and classical communication.

Without loss of generality, by relabeling the qubit indices, suppose that the first nontrivial rank-one projective measurement $\{ \ket{v_0}\bra{v_0}, \ket{v_1}\bra{v_1} \}$ with $\ket{v_0}\bra{v_0} + \ket{v_1}\bra{v_1} = \mathbb{1}$ acts on the first qubit. We rewrite the $n$-qubit target state $\ket{\psi_k^{(n)}}:=\ket{\psi_k}$ as
\begin{equation}
    \ket{\psi_k^{(n)}} =
    \ket{\phi_{k,1}}\otimes\ket{\psi_k^{(n-1)}},
\end{equation}
where $\ket{\psi_k^{(n-1)}} = \bigotimes_{j=2}^n\ket{\phi_{k,j}}$ is the remaining $(n-1)$-qubit state.

For labels $k$ having nonzero probability of yielding outcome $a$, i.e., $\langle v_a|\phi_{k,1}\rangle\neq0$, the corresponding post-measurement states $ \langle v_a|\psi_k^{(n)}\rangle = \langle v_a|\phi_{k,1}\rangle \ket{\psi_k^{(n-1)}} \propto \ket{\psi_k^{(n-1)}} $ must remain orthogonal, since they must be perfectly distinguished by the subsequent LOCC measurements. In other words, the two subsets ${\cal E}^{(n-1)}_{a} := \left\{ \ket{\psi_k^{(n-1)}}: \langle v_a|\phi_{k,1}\rangle \neq 0 \right\}$ for $a=0,1$ are both perfectly distinguishable by LOCC. We then have the following two cases:

\begin{enumerate}
	\item For a Pauli measurement, i.e., $\ket{v_a}\bra{v_a}=\Pi_a^P$ for some nonidentity Hermitian single-qubit Pauli operator $P\in{\cal P}_1$, there may be stabilizer states $\ket{\phi_{k,1}}$ such that $\langle v_a|\phi_{k,1}\rangle=0$. The corresponding label can therefore be ruled out from ${\cal E}^{(n-1)}_{a}$ when outcome $a$ is observed. One can then continue the same procedure recursively, conditioned on the measurement outcome $a$, for the remaining ensemble ${\cal E}^{(n-1)}_{a}$.
	\item If $\{\ket{v_0}\bra{v_0},\ket{v_1}\bra{v_1}\}$ is not a Pauli measurement, then $\langle v_a|\phi_{k,1}\rangle \neq 0$ for every single-qubit stabilizer state $\ket{\phi_{k,1}}$ and both outcomes $a=0,1$. Then, the full reduced ensemble ${\cal E}_{0}^{(n-1)} = {\cal E}_{1}^{(n-1)} := \{ \psi_k^{(n-1)} \}_{k=1}^K$ must be perfectly distinguishable by the subsequent LOCC measurements for either outcome. Thus, the non-Pauli measurement can be skipped, and one can continue on the same procedure directly to the remaining ensemble ${\cal E}^{(n-1)}_{a}$.
\end{enumerate}

In every branch, the remaining states are products of single-qubit stabilizer states and are mutually orthogonal within that branch. Repeating the same argument recursively removes every non-Pauli measurement while retaining the Pauli measurements. Since one qubit is removed at each step, the procedure terminates and yields a perfect discrimination protocol consisting only of adaptive local Pauli measurements, which is an SO protocol. Thus, perfect finite-round LOCC discrimination implies perfect SO discrimination for every ensemble satisfying the assumptions of the theorem. Taking the contrapositive completes the proof.
\end{proof}

\end{appendix}

\end{document}